\documentclass[sigplan, screen]{acmart}
\renewcommand\footnotetextcopyrightpermission[1]{}
\AtBeginDocument{%
  \providecommand\BibTeX{{%
    \normalfont B\kern-0.5em{\scshape i\kern-0.25em b}\kern-0.8em\TeX}}}

\usepackage{tikz}
\usepackage{amsmath}
\DeclareMathAlphabet{\mathcal}{OMS}{cmsy}{m}{n}
\usepackage{textgreek}
\usepackage[english]{babel}
\usepackage{amsthm}
\newtheorem{lemma}{Lemma}
\newtheorem{theorem}{Theorem}
\usepackage{qtree}
\usepackage[rightcaption]{sidecap} 
\usepackage{listings}
\usepackage{xurl}
\usepackage{nameref}

\def\circle<#1>{\raisebox{.5pt}{\textcircled{\raisebox{-.9pt} {#1}}}}

\usepackage[neverdecrease]{paralist}
\usepackage{adjustbox}
\usepackage{subcaption}
\usepackage{xspace}
\newcommand{\algo}{V-Guard\xspace}

\begin{document}
\title{V-Guard: An Efficient Permissioned Blockchain for Achieving Consensus under Dynamic Memberships}

\author{Gengrui Zhang, Yunhao Mao, Shiquan Zhang, Shashank Motepalli, and Hans-Arno Jacobsen}
\affiliation{%
  \institution{\textit{University of Toronto}}
}

\begin{abstract}
This paper presents \algo, a new permissioned blockchain that achieves consensus for vehicular data under changing memberships, targeting the problem in V2X networks where vehicles are often intermittently connected on the roads. 
To achieve this goal, \algo integrates membership management into the consensus process for agreeing on data entries. It binds a data entry with a membership configuration profile that describes responsible vehicles for achieving consensus for the data entry. As such, \algo produces chained consensus results of both data entries and their residing membership profiles, which enables consensus to be achieved seamlessly under changing memberships. In addition, \algo separates the ordering of transactions from consensus, allowing concurrent ordering instances and periodic consensus instances to order and commit data entries. These features make \algo efficient for achieving consensus under dynamic memberships with high throughput and latency performance.
\end{abstract}

\maketitle

\section{Introduction}
The widespread adoption of vehicular automation technologies in modern vehicles has resulted in automobile manufacturers becoming centralized data monopolies. These manufacturers gather vast amounts of data generated by consumers' vehicles, but the data is often inaccessible to users, who must obtain permission from the manufacturer to access it~\cite{teslashprotest1, teslashprotest2, yourcaryourdatayourchoice}. The lack of transparency and control over vehicle operating data has become an increasing concern, particularly with the rise of reports and investigations into malfunctions of immature automation technologies~\cite{carinsurance, nhtsareport, teslascrahes, telsadrivelesscrash, waymoscrashes}. For example, it is difficult to determine the extent of driver responsibility in the event of an accident when only one source of data comes from the monopoly~\cite{americanbar, pattinson2020legal}. Unsurprisingly, recent surveys show that consumers consider data integrity and transparency as their top concerns regarding automation technologies~\cite{huff2021tell, abraham2016autonomous}.
Moreover, insurers and lawmakers are concerned about the automotive industry's data management practices. Insurers need vehicle operation and safety data to improve their risk models and pricing~\cite{insurer1, insurer2} while lawmakers seek to ensure accountability and safety~\cite{legislation, americanbar}.

In addition, the preservation of consumers' right to data ownership has become the focus of numerous regulations in recent years. For example, the European Union's General Data Protection Regulation (GDPR) and the California Privacy Rights Act (CPRA) share similar principles for data management, including data integrity, transparency, ownership, and accountability~\cite{gdpr}. Consequently, there has been a growing chorus of voices demanding more transparency and independent validation of vehicle data~\cite{forbesprivacy, auotnewstrans}.

\subsection{Why a blockchain?}
The currently used centralized solutions, such as manufacturer data collectors and event data recorders, often fall short of upholding the principles of transparent data management and data ownership preservation as outlined in the GDPR and CPRA. First, they are vulnerable to manipulation since they rely on a single source of truth (the manufacturers). Second, users are dependent on the manufacturer's permission to access their own data. Third, the need for additional hardware equipment may make these solutions less cost-effective~\cite{blackbox2, blackbox101}.

In contrast, blockchain solutions, due to their decentralized nature, inherently disrupt data monopolies and eliminate the potential for manipulation from a singular source of "truth". They can reinforce the principles of preserving data ownership and ensuring transparency~\cite{gmblockchainforbes, clavin2020blockchains}. 
Specifically, blockchain solutions in vehicle-to-everything (V2X) networks can circumvent the aforementioned drawbacks. V2X networks allow vehicles to communicate via wireless communication technologies, such as Dedicated Short-Range Communication (DSRC) and Cellular Vehicle-to-Everything (C-V2X). As such, data can be immutably recorded using joint consensus, which prevents both users and manufacturers from fabricating or erasing evidence. Since vehicles are manufactured with identities~\cite{vin}, permissioned blockchains, which scrutinize the identity of their participants, can be applied to achieve deterministic consensus with high performance~\cite{ibmblockchain, limechain}. Furthermore, employing blockchain solutions as independent software on the operating systems of vehicles results in lower product and maintenance costs compared to using event data recorders. 
These features provide users with greater control over their own data, while also promoting accountability and safety within the industry~\cite{mckinseyblockchain}. 

\subsection{Why a new blockchain?}
Blockchain initiatives have been launched by leading automobile manufacturers, such as BMW~\cite{bmwbc}, Mercedes~\cite{mercedesbc, forbesmercedesbc}, Volkswagen~\cite{vwbc1, vwbc2}, Toyota~\cite{toyotabc}, and Ford~\cite{forbesmercedesbc1, forbesmercedesbc2}. 
However, these applications use traditional permissioned blockchains (e.g., HyperLedger Fabric~\cite{androulaki2018hyperledger}, CCF~\cite{russinovich2019ccf}, and Diem~\cite{diem2020}) and primarily focus on supply chain management, ridesharing, and carsharing~\cite{cnbcbc, forbesbcautomotive, qadbcautomotive}. The consensus algorithms (e.g., PBFT~\cite{castro1999practical} and HotStuff~\cite{yin2019hotstuff}) of these blockchains are specifically designed to function in a stable environment with a relatively small volume of data. Each server in the network is configured with a predefined membership profile that typically consists of a static set of servers. When membership changes occur (e.g., new servers join or old servers leave), the consensus algorithms halt ongoing services and update the membership profile on each server using additional reconfiguration approaches (e.g., MC~\cite{rodrigues2010automatic} and Lazarus~\cite{garcia2019lazarus}) to consistently manage membership changes~\cite{lamport2010reconfiguring, martin2004framework}. 

This suspend-and-update method, which creates a ``hard stop'' in transaction consensus, may be acceptable when membership changes are infrequent. However, this method is not suitable for blockchains utilized in V2X, where a \emph{dynamic environment with frequent membership changes} caused by vehicles' arbitrary connectivity is the common place~\cite{garcia2021tutorial, v2xconnectivity, qualcommv2x}. For example, in DSRC V2X networks, the vehicles in a running vehicle's vicinity can change dynamically on the road, causing them to go online and offline depending on the distance. Similarly, in Cellular V2X networks, it is difficult to predict which drivers will power off their vehicles, resulting in participants becoming unavailable.

To cope with frequent membership changes, traditional BFT consensus algorithms (e.g., PBFT~\cite{castro1999practical} and HotStuff~\cite{yin2019hotstuff}) will suffer from significant performance degradation~\cite{duan2022foundations, lorch2006smart}. For example, an MC~\cite{rodrigues2010automatic} reconfiguration takes around $10$-$20$ seconds to complete a membership change. The additional costs incurred by membership management can become significant under high system workloads and large scales~\cite{abraham2016solida, rodrigues2010automatic, bessani2020byzantine}. 
Therefore, new solutions are urgently needed for V2X blockchains that can achieve consensus seamlessly during membership changes without impeding ongoing consensus.

\subsection{\algo for membership dynamicity}

To address the unique challenges presented by V2X networks, we introduce \algo, a permissioned blockchain architecture that utilizes a novel consensus algorithm. \algo targets the challenge of achieving efficient consensus in the face of dynamically changing memberships.
Unlike conventional BFT algorithms (e.g., PBFT~\cite{castro1999practical} and HotStuff~\cite{yin2019hotstuff}) that focus solely on achieving consensus for data transactions, \algo incorporates membership configuration agreement into the consensus of data transactions. It provides a traceable and immutable ledger for not only data transactions but also for their corresponding membership configurations. As such, the consensus target of \algo becomes: 

\texttt{<data transactions, membership configurations>}

\algo employs a membership management unit (MMU) to handle dynamic membership configurations. The MMU maintains a record of all available vehicles that can act as potential members and groups them into various sets. It generates a queue of membership configuration profiles, known as booths, and provides an interface to the consensus services. During the consensus process, a data entry is paired with a designated booth in both the ordering and consensus instances. When the current booth becomes unavailable, the ordering/consensus instance obtains a new booth from the MMU. Thus, \algo achieves consensus not only on what data is committed but also on which booth (vehicles) committed the data. 

In addition, since vehicle data produced by different sensors is often transactionally unrelated, \algo allows concurrent ordering for data by separating ordering from consensus, where the ordering and consensus phases can take place in different booths. The ordering phase concurrently appends data entries to a total order log, and the consensus phase periodically commits the ordered entries to an immutable ledger. This separation minimizes message passing, allowing \algo to achieve high throughput. Furthermore, the separation empowers \algo to manage the memberships in a more granular way; it enables data entries to be ordered and committed in different booths; that is, consensus can be achieved seamlessly in dynamic memberships.

We implemented \algo and evaluated its performance under both static and dynamic memberships. \algo achieved high consensus throughput and low latency. For example, with a membership size of $n=4$, \algo achieved a peak consensus throughput of $765,930$~TPS (transactions per second) at a latency of $143$~ms under static membership and $615,328$ TPS at $1324$~ms under dynamic memberships.
\algo aims to record and replicate only critical data reflecting the status and decisions of autonomous driving systems, as the data that requires recording is relatively small. This includes essential variables such as speed, acceleration, direction, and object detection, rather than the raw data. Therefore, the achieved performance shows that \algo is a practical and feasible solution.

To summarize, \algo has the following novel features.

\begin{itemize}

    \item It achieves consensus seamlessly during membership changes without impeding transaction consensus. It integrates membership management into each phase of the consensus process, targeting the problem of frequent membership changes in blockchains for V2X networks.
        
    \item It separates ordering from consensus and thus enables concurrent data ordering, reducing ordering latency and messaging costs. The separation also empowers data entries to be ordered (ordering) and committed (consensus) in different memberships.
     
    \item Its implementation achieves high performance in terms of throughput and latency in both static and dynamic memberships. The implementation is open-source and available at \url{https://github.com/vguardbc/vguardbft}.
    
\end{itemize}

\section{V-Guard overview}

\algo provides a blockchain service among vehicles and their manufacturers to avoid a single point of control of data. To apply this service, a vehicle creates a \algo instance and starts to operate as a \emph{proposer}. The vehicle can also join other vehicles' instances by operating as a \emph{validator}. Each \algo instance must include at least four members: the proposer, automobile manufacturer (a.k.a., pivot validator), and other (at least two) vehicles (a.k.a., vehicle validators). The description of the roles is presented below.

\begin{description}
    \item[Proposer ($V_p$).] Each vehicle is the proposer of its own \algo instance, so \emph{each instance has only one proposer}. The proposer operates as a leader in its \algo instance: it proposes data to validators and coordinates consensus for proposed data among validators.
    
    \item[Validator.] In a \algo instance, except the proposer, other members operate as a validator. The automobile manufacturer always operates as a \textbf{pivot validator} (denoted by $V_{\pi}$), and the other vehicles operate as \textbf{vehicle validators} (denoted by $V_i$).
    
\end{description}

A vehicle may participate in multiple \algo instances simultaneously: it operates as the proposer of its own \algo instance and as a validator in other vehicles' \algo instances. We denote the number of instances a vehicle joins as \emph{catering factor}, $\gamma$.

\begin{description}
    \item[Catering factor ($\gamma$).] The catering factor ($\gamma \in \mathbb{Z}^+$) of a vehicle describes the number of \algo instances this vehicle is currently participating in. It changes correspondingly when the vehicle joins and quits other vehicles' \algo instances.
\end{description}

\begin{figure}[t]
    \centering
    \includegraphics[width=0.95\linewidth]{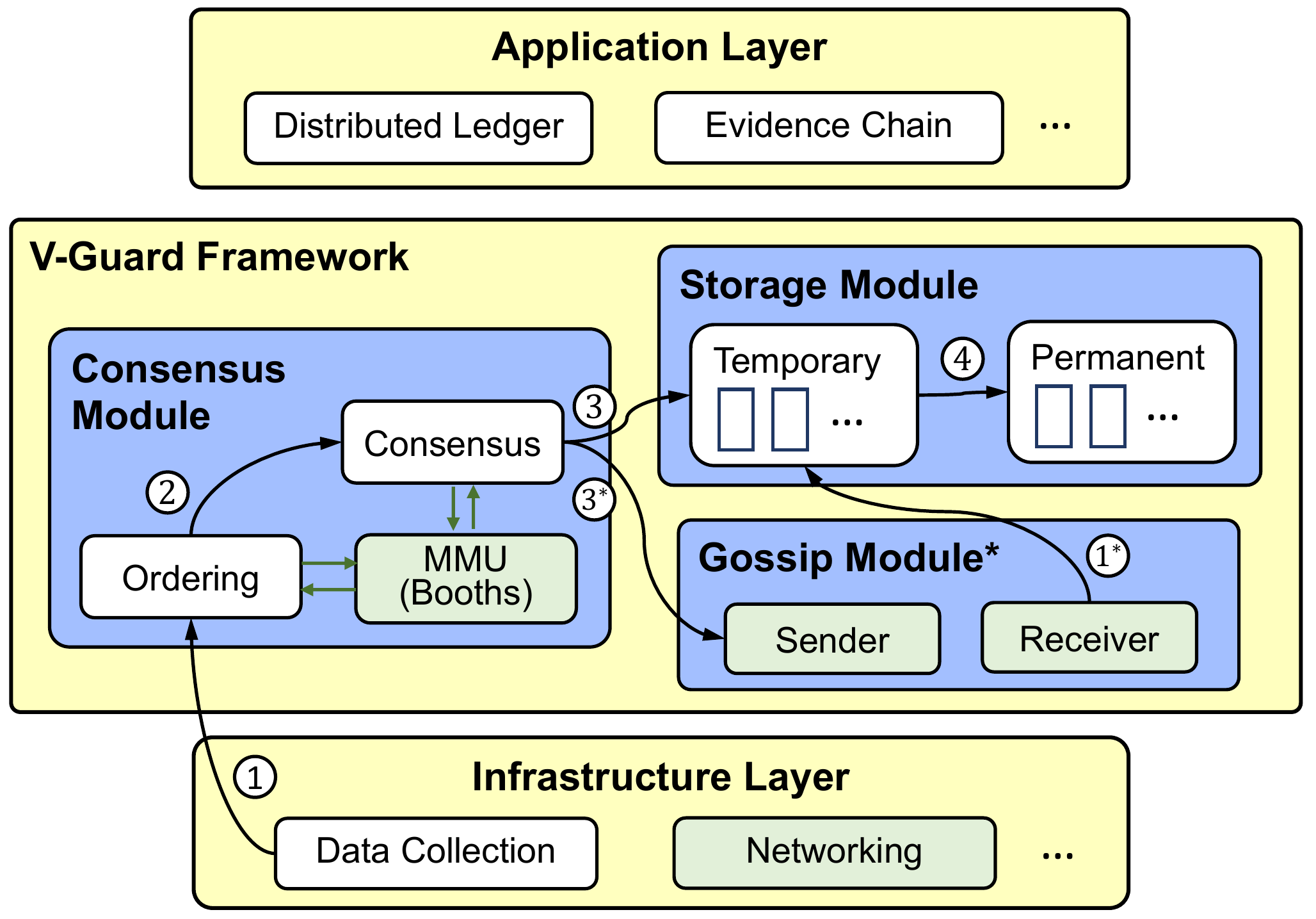}
    \caption{\algo architecture. The consensus module  operates in ordering and consensus instances interacting with the MMU. Committed data are stored in the storage module and further disseminated through the gossip module.}
    \label{fig:vgframework}
\end{figure}

\algo operates as a middleware system that is independent of its residing automotive operating system. It takes input from the vehicle's infrastructure layer for data collection (e.g., from sensors and monitoring devices), communicates with other vehicles, and provides a blockchain service that underpins distributed applications in V2X networks. In general, as shown in Figure~\ref{fig:vgframework}, a \algo instance consists of two core modules of consensus and storage and an optional module of gossiping (not required for correctness). 

\begin{itemize}
    \item The consensus module coordinates consensus for collected data entries (\S\ref{sec:consensus-module}). The proposer uses its membership management unit (MMU) to obtain a valid membership (\S\ref{sec:boothcomposer}). Then, it starts an ordering instance to coordinate with other validators in the obtained configuration and secures a unique ordering ID for the data entries (\S\ref{sec:ordering}). When a consensus instance, periodically scheduled with its own membership configuration, commits the ordered entries, the consensus is achieved (\S\ref{sec:consensus}).

    \item The gossip module is pluggable, with no effect on system correctness (\S\ref{sec:gossiping}). It aims to increase robustness by further disseminating and replicating committed data entries to more vehicles in the network.
        
    \item The storage module stores data entries committed in the consensus module (\S\ref{sec:storage}). It temporarily stores all entries by default according to a predefined policy (e.g., for $24$ hours) and provides an option for users or applications to permanently store selected data. 
        
\end{itemize}

\section{System model}
\label{sec:system-model}

V-Guard is a permissioned blockchain, which requires identity information from members; i.e., each vehicle can identify who they are communicating with. Since vehicles are manufactured with identities\footnote{Vehicle identification numbers (VIN)~\cite{vin} serve as vehicles' fingerprint, as no two vehicles in operation have the same VIN.}, this standard practice of permissioned blockchains is viable in practice. 

\textit{Service properties.} \algo's blockchain architecture provides evidence-based assurance that vehicles behave as intended. For example, it can be used to record the decision-making data produced by applied automation technologies. \textbf{\algo provides two guarantees: no party can \circle<1> fabricate evidence or \circle<2> bury evidence.} 
Data entries are ordered and committed by quorums of independent vehicles with their signatures. After consensus is reached, data entries are traceable and immutable.
Furthermore, neither the proposer nor the consensus pivot can bury evidence by claiming data loss because each committed data entry is consistently stored in vehicles included in the consensus process.

\emph{Network assumptions.}
\algo does not rely on network synchrony to ensure safety but requires \emph{partial synchrony} for liveness; i.e., message delivery and message processing times between two vehicles are bounded by a delay after an unknown global stabilization time (GST). The network may fail to deliver messages, delay them, duplicate them, or deliver them out of order.

\emph{Fault tolerance.} \algo assumes a Byzantine failure model where faulty members may behave arbitrarily, subject only to independent member failures. Each member is an independent entity, as \algo runs independently regardless of its underlying operating system; that is, when $f$ vehicles, regardless of their manufacturers, are faulty (individual or colluding), they are counted as $f$ failures. 
Note that a vehicle's behavior may not be universal when participating in multiple \algo instances ($\gamma >1$); i.e., it could operate correctly in one instance and falsely in another one.

In addition, we do not consider junk-pouring attacks (e.g., DDoS attacks) from faulty members, as such behavior can not be handled by consensus algorithms alone. They can often be handled at lower levers, such as rate controls and blacklist methods.

\emph{Cryptographic primitives.} \algo applies ($t$, $n$)-threshold signatures where $t$ out of $n$ members can collectively sign a message~\cite{shoup2000practical, libert2016born}. Threshold signatures can convert $t$ partially signed messages (of size $O(n)$) into one fully signed message (of size $O(1)$). The fully signed message can then be verified by all $n$ members, proving that $t$ members have signed it. The use of threshold signatures to obtain linearity has been common in state-of-the-art linear BFT protocols, such as HotStuff~\cite{yin2019hotstuff}, SBFT~\cite{gueta2019sbft}, and Prosecutor~\cite{zhang2021prosecutor}.

\emph{A note on changing members.} Similar to other BFT protocols, \algo requires a minimum quorum size of $2f+1$ to tolerate $f$ Byzantine failures in a total of $n = 3f+1$ members. Since \algo is designed to provide an evidence chain for vehicles and their manufacturers, we always include the vehicle's manufacturer (i.e., pivot validator) as a member during membership changes; thus, a new booth can have up to $n-2$ new members, and both the proposer and pivot validator can witness the agreement of each data entry. 

\section{The V-Guard consensus protocol}
\label{sec:consensus-module}
This section presents the \algo consensus protocol in detail. As shown in Figure~\ref{fig:vguard-consensus}, \algo has two types of services: \emph{event-driven} and \emph{daemon} services. Event-driven services are invoked when data is collected. Their orchestration achieves consensus for the collected data and forms a blockchain recording data and membership information. Daemon services handle network connections by dynamically connecting to and disconnecting from vehicles according to their availability.

\subsection{Membership management unit (MMU)}
\label{sec:boothcomposer}

The MMU keeps track of available vehicle connections in the network and prepares a queue of valid \textit{booths}; examples with animations can be found on \algo's GitHub page~\cite{vguardsource}). A booth, denoted by $\mathcal{V}$, is a \emph{membership configuration profile} that describes a set of participating members (vehicles). The profile contains a list of vehicle information on their network addresses (i.e., \texttt{v.conn}) and public keys (i.e., \texttt{v.pub}). The structure of a booth is listed below,
{\small
    \begin{verbatim}
    var boothQueue = struct {
        sync.RWMutex
        b []Booth //the queue of available booths
    }
    
    type Booth struct{
        v []struct{
            conn net.Conn //connection information
            pub  share.PubPoly //individual public key
            ...
        }
        thpub share.PubPoly //threshold signature public key 
        ...
    }
    \end{verbatim}  
}
\vspace{-1em}
The size of each booth is a predefined parameter established in the system setup. A valid booth must include the proposer ($V_p$), pivot validator ($V_{\pi}$), and the other vehicles as vehicle validators. Since \algo assumes BFT failures, the minimum size of a booth is four ($3f{+}1$) to tolerate $f{=}1$ Byzantine failure~\cite{fischer1986easy}.

The MMU interface provides available booths for ordering and consensus instances (denoted by $\mathcal{V}_o$ and $\mathcal{V}_c$, respectively). If there is no booth available in the queue (e.g., when the vehicle has no network connection in a tunnel), the ordering and consensus instances wait for new booths from the interface. In addition, the MMU dynamically manages booths based on their availability and efficiency. It removes a booth when $f$ vehicles are unavailable and periodically pings the other vehicles, prioritizing the booth with the lowest network latency.

In each booth, the proposer first establishes the distributed key generation (DKG) of a ($t$, $n$) threshold signature, where $t$ is the quorum size, and $n$ is the booth size. Once the DKG is completed, all vehicles within the same booth share the same public key for threshold signatures (\texttt{thpub}). It is important to mention that while vehicles from outside the booth cannot verify this threshold signature's public key, they can still verify the individual public key of a vehicle  (i.e., \texttt{v.pub}) as all vehicles use the same \algo key generation service.

\begin{figure}[t]
    \centering
    \includegraphics[width=0.95\linewidth]{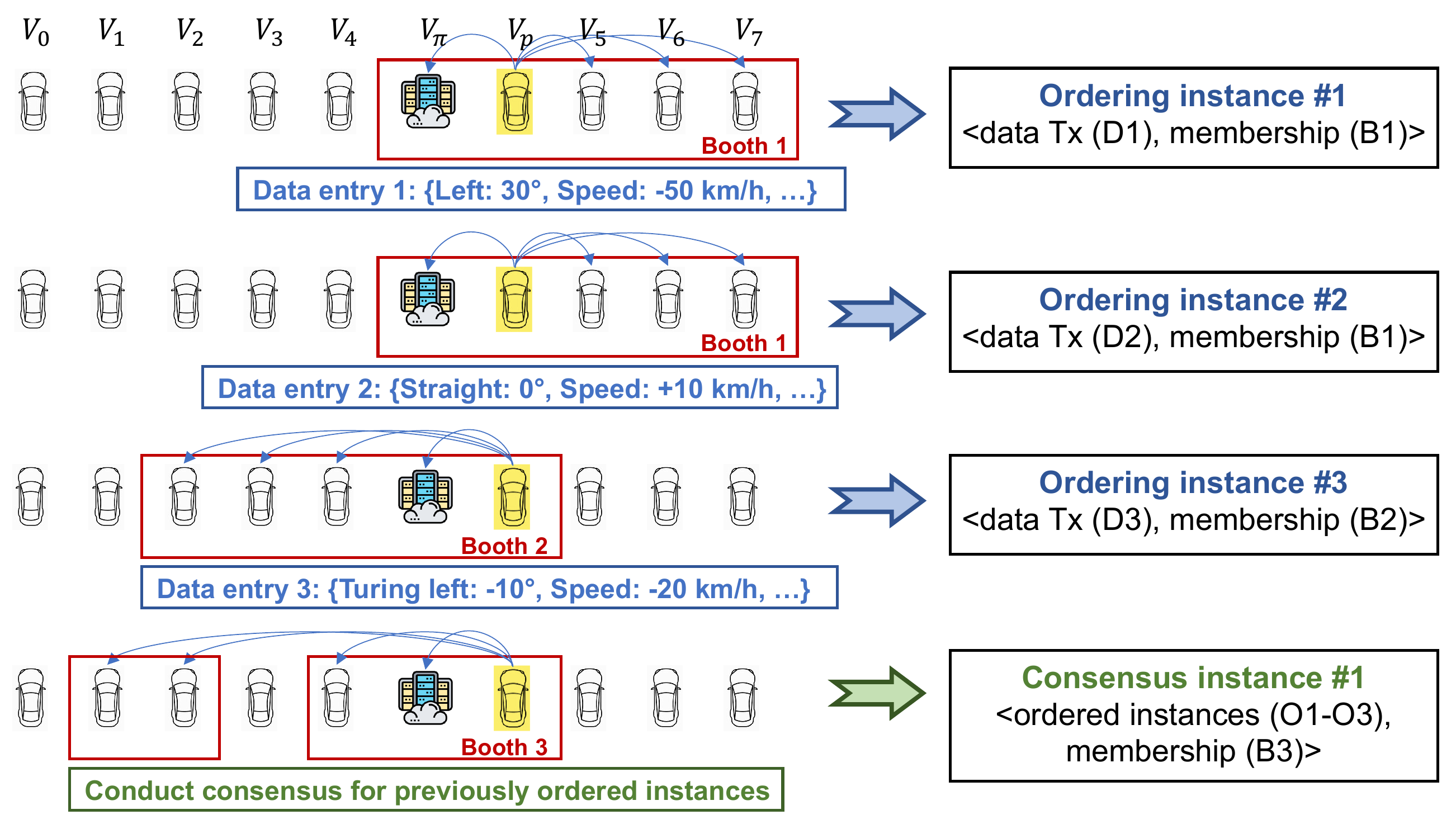}
    \caption{An example of MMU-embedded ordering and consensus procedures. When the current booth is no longer valid for ordering or consensus instances, the proposer ($V_p$) retrieves a new booth from the MMU. The agreement of membership profiles is paired with the agreement of data transactions in each ordering and consensus instance.}
    \label{fig:mmuexamples}
\end{figure}

With the feature of MMU, \algo can issues ordering and consensus procedures in changing memberships.
For example, Figure~\ref{fig:mmuexamples} shows a \algo instance of $10$ participating vehicles, where $V_p$ is the proposer, $V_{\pi}$ is the pivot validator, and $V_0$ to $V_7$ are $8$ vehicle validators. In this case, when the proposer issues an ordering instance for data entry 1, the procedure takes place in Booth 1 (the initial booth). The next ordering instance (i.e., data entry 2) will reuse the current booth (Booth 1) if it is still available. When the current booth becomes unavailable (e.g., some vehicles go offline), the MMU will provide a new booth (Booth 2) for the next ordering instance (i.e., data entry 3). As such, each ordering instance produces a paired result of an ordered data transaction and its residing membership information (details in \S\ref{sec:ordering}). When a scheduled consensus instance takes place, it interacts with the MMU in the same behavior to reuse the current booth or get a new booth. The consensus instance sustains the output of the ordering instances by committing the paired results, thereby ensuring the immutability of these results across different memberships. Equipped with the coalition between the MMU and ordering/consensus instances, \algo can achieve consensus seamlessly under changing members and produce an immutable ledger recording traceable data transactions with their corresponding membership profiles.

\begin{figure*}[t]
    \centering
    \includegraphics[width=\linewidth]{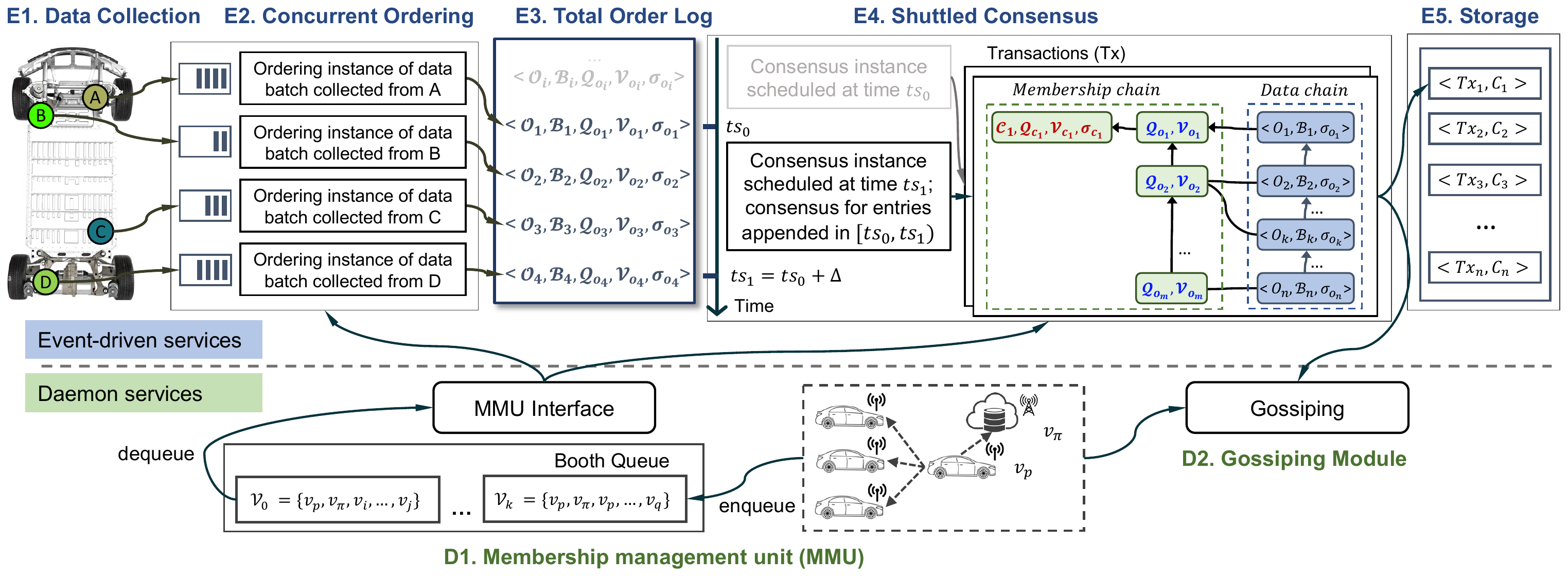}
    \caption{\algo workflow. The MMU prepares a queue of booths for ordering and consensus instances. Ordering instances concurrently append data entries to a totally ordered log. Consensus instances aggregate ordered entries into transactions and commit them periodically.
    Committed transactions can be further disseminated to other vehicles via gossiping.}
    \label{fig:vguard-consensus}
\end{figure*}

\subsection{Ordering -- Form a total order log}
\label{sec:ordering}
The ordering phase secures a unique sequence number (order) for each data batch on a total order log. \algo separates the ordering of data batches from consensus. The separation enables a non-blocking ordering procedure in which the ordering of a data batch does not wait for the completion of the consensus of its prior batch. The separation also allows ordering and consensus to be conducted in different booths during membership changes. The two features make \algo more applicable and efficient in V2X networks where vehicles are connected over unstable networks.

\textbf{The ordering phase instances append log entries to a total order log showing that \circle<1> what data is agreed upon in \circle<2> what order and by \circle<3> what vehicle membership.} 
Specifically, first, the proposer starts an ordering instance when collected data reaches their predefined batch sizes (denoted by $\beta$). Then, the ordering instance distributes data batches collected from the proposer to all validators, including the pivot and vehicles (``what data''); the proposer assigns each data batch with a unique ordering ID (i.e., the sequence number) (``what order'') and a membership configuration that indicates which validators agree on the ordering ID as a quorum (``what membership''). We describe the workflow of an ordering phase instance for a data batch as follows.

\setdefaultleftmargin{0.6cm}{0.6cm}{0.3cm}{}{}{}
\begin{enumerate}
\renewcommand{\labelenumi}{\textbf{O\arabic{enumi}}}

    \item \label{ordering:1}
    The proposer $V_p$ assigns a unique sequence number (denoted by $i$) as the ordering ID (denoted by $O_i$) to a data batch (denoted by $\mathcal{B}_i$) and starts an ordering instance.
    \begin{enumerate}
        \item $V_p$ confirms with the MMU that the booth for this ordering instance ($\mathcal{V}_o$) is available and up-to-date; if not, it invokes the MMU to get a new $\mathcal{V}_o$.
        
        \item $V_p$ sends $\langle \texttt{Pre-Order}, \mathcal{B}_i, h_{\mathcal{B}_i}, \mathcal{V}_o,  h_{\mathcal{V}_o}, i, \sigma_{V_p} \rangle$ to all members in $\mathcal{V}_o$, where $h_{\mathcal{B}_i}$ is the hash of $\mathcal{B}_i$; $h_{\mathcal{V}_o}$ is the hash of $\mathcal{V}_o$; and $\sigma_{V_p}$ is the signature that $V_p$ signs the combination of $i$, $h_{\mathcal{B}_i}$, and $h_{\mathcal{V}_o}$.
        
        \item $V_p$ creates a set $\mathcal{R}_{o_i}$, waiting for replies from validators in $\mathcal{V}_o$.
    \end{enumerate}
    
    \item \label{ordering:2}
    After receiving a \texttt{Pre-Order} message, $V_i$ first verifies the message and then replies to $V_p$. The verification succeeds when \circle<1> $h_B$ and $h_{\mathcal{V}_o}$ match the hashes of $B$ and $\mathcal{V}_o$, respectively; \circle<2> $\sigma_{V_p}$ is valid; and \circle<3> $i$ has not been used. If the verification succeeds, $V_i$ then sends $\langle \texttt{PO-Reply}, i, \sigma_{V_i} \rangle$ to $V_p$, where $\sigma_{V_i}$ is the signature that $V_i$ signs the combination of $i$, $h_B$, and $h_{\mathcal{V}_o}$.
    
    \item \label{ordering:3}
    $V_p$ collects \texttt{PO-Reply} messages from validators and adds them to $\mathcal{R}_{o_i}$ until it receives $2f$ replies (i.e., $|\mathcal{R}_{o_i}| = 2f$); then, $V_p$ takes the following actions.
    
    \begin{enumerate}
        \item $V_p$ converts the $2f$ collected signatures (i.e., $\sigma_{V_i}$ in $\mathcal{R}_{o_i}$) to a ($t$, $n$)-threshold signature, $\sigma_{o_i}$, where $t=2f$.
        
        \item $V_p$ creates a subset $\mathcal{Q}_{o_i}$ including the $2f$ vehicles, the signatures of which are converted to $\sigma_{o_i}$, from $\mathcal{V}_o$ as a membership quorum; i.e., $|\mathcal{Q}_{o_i}|=2f \land \mathcal{Q}_{o_i} \subset \mathcal{V}_o$. 
        
        \item $V_p$ sends $\langle \texttt{Order}, i, \mathcal{Q}_{o_i}, \sigma_{o_i} \rangle$ to all members in $\mathcal{V}_o$.
    \end{enumerate}
    
    $V_p$ now considers batch $\mathcal{B}_i$ ordered with ordering ID $O_i$ ($O_i=i$) by the quorum $\mathcal{Q}_{o_i}$ of the members from $\mathcal{V}_o$. It appends $\langle O_i, \mathcal{B}_i, \mathcal{Q}_{o_i}, \mathcal{V}_o, \sigma_{o_i} \rangle$ to the total order log.
    
    \item \label{ordering:4}
    After receiving an \texttt{Order} message, $V_i$ verifies it by three criteria: 
    \circle<1> $\sigma_{o_i}$ is valid with a threshold of $2f$ (correct quorum sizes);
    \circle<2> signers of $\sigma_{o_i}$ match $\mathcal{Q}_{o_i}$ (valid validators); 
    and \circle<3> $\forall$ $V_i \in \mathcal{Q}_{o_i}, V_i \in \mathcal{V}_o$ (the same membership).
    
    If the validation succeeds, $V_i$ considers batch $\mathcal{B}_i$ ordered with ID $O_i$ by quorum $\mathcal{Q}_{o_i}$ in $\mathcal{V}_o$. It appends $\langle O_i, \mathcal{B}_i, \mathcal{Q}_{o_i}, \mathcal{V}_o, \sigma_{o_i} \rangle$ to the total order log.

\end{enumerate}

In each ordering instance, each data entry is paired up with a verifiable booth. Validators receive the booth information ($\mathcal{V}_o$) in \textcolor{black}{O}\ref{ordering:1} and verify the quorum based on the received booth in \textcolor{black}{O}\ref{ordering:4}. The paired information is ordered with a unique ordering ID and will be committed in a consensus instance.

\subsection{Consensus -- Commit the total order log}
\label{sec:consensus}
The consensus phase ensures that all correct members in a booth agree on the entries shown on the total order log and produces an immutable ledger that periodically commits these log entries.

\textbf{Shuttled consensus.}
\algo periodically initiates consensus instances to commit the log entries appended to the total order log. Each consensus instance, similar to a ``shuttle bus'', is initiated at regular time intervals, denoted by $\Delta$ (e.g., $\Delta =100$~ms), and is responsible for log entries appended in $\Delta$. Formally, a consensus instance scheduled at time $ts+\Delta$ is responsible for log entries appended in [$ts$, $ts+\Delta$), where initially $ts=0$. Thus, consensus instances cover all the entries on the total order log; i.e., \textbf{no two adjacent consensus instances leave uncommitted log entries in between.} We use the starting timestamp (i.e., $ts$ in [$ts$, $ts+\Delta$)) as the ID of the corresponding consensus instance (i.e., $\mathcal{C}_i = ts$).

\textbf{Membership pruning.} 
We denote all the entries to be committed in a consensus instance $\mathcal{C}_i$ as a transaction ($tx$); i.e., $tx = \lbrace entry_{ts}, ..., entry_{ts+\Delta} \rbrace$. A transaction may contain multiple log entries whose orderings were conducted in the same membership. In this case, the transaction prunes away redundant membership and quorum information, linking entries ordered in the same booth to the same membership profile. For example, in Figure~\ref{fig:vguard-consensus}~\textcolor{black}{E4}, if entries $\lbrace \langle \mathcal{O}_2, \mathcal{B}_2 \rangle, ..., \langle \mathcal{O}_k, \mathcal{B}_k \rangle \rbrace$ are ordered by the membership configuration $\langle \mathcal{Q}_{O_2}, \mathcal{V}_{O_2}, \rangle$, then they link to only one membership profile. When membership changes infrequently, pruning can reduce the message size while maintaining the traceable membership feature. After applying pruning, a consensus instance with its ID $\mathcal{C}_i = ts$ (i.e., scheduled at time $ts+\Delta$) works as follows.

\begin{enumerate}
\renewcommand{\labelenumi}{\textbf{C\arabic{enumi}}}
    
    \item \label{c:c1}
    $V_p$ prepares a pruned transaction ($tx$) for log entries appended in $[ts, ts+\Delta)$. 
    
    \begin{enumerate}
        \item $V_p$ confirms with the MMU that the booth for this consensus instance ($\mathcal{V}_c$) is available and up-to-date. Note that $\mathcal{V}_c = \mathcal{V}_o$ if $\mathcal{V}_o$ is still available; otherwise, $\mathcal{V}_c$ is a new booth provided by the MMU. 
        
        \item 
        $V_p$ sends a \texttt{Pre-Commit} message to all $V_i$ in $\mathcal{V}_c$. Since the consensus phase may take place in different booths (i.e., $\mathcal{V}_c \neq \mathcal{V}_o$), $V_p$ checks if $V_i$ was previously in the quorums and saw the entries in $tx$.  
        \begin{itemize}
            \item If $V_i$ has seen all entries in $tx$, $V_p$ sends the starting and ending ordering IDs ($O_{ts}$ and $O_{ts+\Delta}$) of entries in $tx$ so that $V_i$ can locate $tx$ from its log. $V_p$ sends 
            \vspace{-0.5em}
            $$\langle \texttt{Pre-Commit}, ts, h_{tx}, O_{ts}, O_{ts+\Delta}, \mathcal{V}_c, h_{\mathcal{V}_c}, \sigma_{V_p} \rangle,$$
            where $h_{tx}$ and $h_{\mathcal{V}_c}$ are the hashes of $tx$ and $\mathcal{V}_c$, respectively, and $\sigma_{V_p}$ is the signature that $V_p$ signs the combination of $ts$, $h_{tx}$, and $h_{\mathcal{V}_c}$.
            
            \item If $V_i$ has not seen the entries in $tx$, $V_p$ sends 
            \vspace{-0.5em}
            $$\langle \texttt{Pre-Commit}, ts, h_{tx}, tx, \mathcal{V}_c, h_{\mathcal{V}_c}, \lbrace\mathcal{R}_{o_i}\rbrace, \lbrace\mathcal{Q}_{o_i}\rbrace, \sigma_{V_p} \rangle,$$
            piggybacking the entries ($tx$) to be committed, where $\lbrace\mathcal{R}_{o_i}\rbrace$ is the set of the signed replies for every entry in $tx$, and $\lbrace\mathcal{Q}_{o_i}\rbrace$ is the set of their corresponding signers (obtained from Step O3 in ordering).
        \end{itemize}
        
        \item $V_p$ creates a set $\mathcal{R}_{c}$, waiting for replies from validators.
    \end{enumerate}
    
    \item \label{c:c2}
    After receiving a \texttt{Pre-Commit} message, $V_i$ first verifies the message based on the following criteria.
    \begin{itemize}
        \item If $V_i$ has seen $tx$, $V_i$ locates the log entries from $O_{ts}$ to $O_{ts+\Delta}$ and calculates the hash of them. \circle<1> The hash must equal to $h_{tx}$, \circle<2> $\sigma_{V_p}$ is valid, and \circle<3> $ts$ has not been used by other consensus instances.
        
        \item If $V_i$ is not in the ordering booth (has not seen $tx$), then $V_i$ is unable to verify the threshold signature because $V_i$ did not participate in the DKG generation process in that booth (discussed in \S\ref{sec:boothcomposer}). Thus, to validate the correctness of the previous ordering phase, it must verify the replies in $\lbrace\mathcal{R}_{o_i}\rbrace$ of each entry included in $tx$. The public keys of the corresponding validators can be found in $\lbrace\mathcal{Q}_{o_i}\rbrace$. In this case,         \circle<1> $V_i$ calculates the hash of $tx$, and the hash must equal to $h_{tx}$; \circle<2> for each data entry  $\in tx$, $2f+1$ signed replies can be found in $\lbrace\mathcal{R}_{o_i}\rbrace$ and the signatures are valid in accordance to $\lbrace\mathcal{Q}_{o_i}\rbrace$; and \circle<3> $ts$ has not been used by previous consensus instances.
    \end{itemize}
    
    If the verification succeeds, $V_i$ sends $\langle \texttt{PC-Reply}, \mathcal{C}_i, \sigma_{V_i} \rangle$ to $V_p$, where $\sigma_{V_i}$ is the signature that $V_i$ signs the combination of $ts$, $h_{tx}$, and $h_{\mathcal{V}_c}$.
    
    \item \label{commitmsg} \label{c:c3}
    $V_p$ collects \texttt{PC-Reply} messages from validators and adds them to $\mathcal{R}_{c}$ until it receives $2f$ replies (i.e., $|\mathcal{R}_{c}|=2f$); $V_p$ then takes the following actions.
    
    \begin{enumerate}
        \item $V_p$ converts the $2f$ collected signatures (i.e., $\sigma_{V_i}$ in $\mathcal{R}_{c}$) to a ($t$, $n$) threshold signature, $\sigma_{c_{ts}}$, where $t=2f$.
        
        \item $V_p$ creates a subset $\mathcal{Q}_{c_{ts}}$ including the $2f$ vehicles, the signatures of which are converted to $\sigma_{c_{ts}}$, from $\mathcal{V}_c$ as a membership quorum; i.e., $|\mathcal{Q}_{c_{ts}}|=2f \land \mathcal{Q}_{c_{ts}} \subset \mathcal{V}_c$. 
        
        \item $V_p$ sends $\langle \texttt{Commit}, ts, \mathcal{Q}_{o_i}, \mathcal{V}_c, \sigma_{c_{ts}} \rangle$ to all members in $\mathcal{V}_c$.        
    \end{enumerate}
        
    $V_p$ now considers transaction $tx$ committed and broadcasts $tx$ to connected vehicles in the gossip module (introduced in~\S~\ref{sec:gossiping}).
    
    \item \label{c:c4}
    After receiving a \texttt{Commit} message, $V_i$ verifies it by three criteria:
    \circle<1> $\sigma_{c_{ts}}$ is valid with a threshold of $2f$;
    \circle<2> signers of $\sigma_{c_{ts}}$ match $\mathcal{Q}_{c_{ts}}$; 
    \circle<3> $\forall$ $V_i \in \mathcal{Q}_{c_{ts}}, V_i \in \mathcal{V}_c$.
    
    If the verification succeeds, $V_i$ considers  $tx$ committed.
\end{enumerate}

With separate ordering and consensus instances, data entries can be ordered and committed in different booths. Nevertheless, in each ordering/consensus instance, the booth is unchangeable (required by safety); i.e., the same booth must persist through \textcolor{black}{O}\ref{ordering:1}-\textcolor{black}{O}\ref{ordering:4} or \textcolor{black}{C}\ref{c:c1}-\textcolor{black}{C}\ref{c:c4}. If a booth fails before an instance completes, the instance aborts the current ordering/consensus process and invokes the MMU to obtain a new booth. Then, it retries the corresponding process.

\textbf{Consistency of memberships.} 
\algo's consensus contains agreements on both data entries and membership profiles. As illustrated in Figure~\ref{fig:vguard-consensus}, committed data batches form a \textit{data chain}, and pruned membership profiles form a \textit{membership chain}. Since each data entry is paired up with its residing instance's membership profiles, membership is agreed upon with data entries; they are both traceable and verifiable through participating members' signatures. This feature of supporting dynamic memberships makes \algo applicable and efficient in V2X networks.

In addition, \algo has a message complexity of $\mathcal{O}(n)$. Steps \textcolor{black}{O}\ref{ordering:1} and \textcolor{black}{O}\ref{ordering:2} in the ordering phase, and \textcolor{black}{C}\ref{c:c1} and \textcolor{black}{C}\ref{c:c2} in the consensus phase have a message complexity of $\mathcal{O}(n)$ as messages of size $\mathcal{O}(1)$ flow only between the proposer and ($n-1$) validators. In Step \textcolor{black}{O}\ref{ordering:3} and \textcolor{black}{C}\ref{c:c3}, the proposer converts $2f$ signatures (in total of size $\mathcal{O}(n)$) to one threshold signature of size $\mathcal{O}(1)$ and broadcasts it to ($n-1$) validators, so the message complexity in this step remains $\mathcal{O}(n)$. Therefore, \algo obtains linear message complexity. 

\subsection{Correctness discussion}
\label{sec:correctness}
We now discuss the correctness of \algo. In contrast to traditional BFT protocols, \algo does not operate in the succession of views, so it does not apply view changes. Instead of relying on views, \algo uses a dynamic membership management system to handle changes in the network. In addition, the proposer takes both the client and the leader roles, as it is the only member that produces data. When the proposer fails, its \algo instance halts, and its validators will join other vehicles' \algo instances when being included by other vehicles' MMUs.

\begin{theorem}[Validity] \label{validity}
 Every data entry committed by correct members in the consensus phase must have been proposed in the ordering phase.
\end{theorem}

\begin{proof}
Only a vehicle can be the proposer of the \algo instance it initiated. In \textcolor{black}{O}\ref{ordering:2}, validators sign the same content as the proposer, which includes the hash of a proposed data entry. In \textcolor{black}{O}\ref{ordering:3} and \textcolor{black}{O}\ref{ordering:4}, a valid threshold signature is converted from a quorum of signatures that sign the same hash of the proposed data entry. In addition, the proposer signs the hash of the transaction that includes the data entry in \textcolor{black}{C}\ref{c:c1}. If a transaction includes non-proposed data entries, the verification process in \textcolor{black}{C}\ref{c:c2} fails. Thus, a committed transaction must include data entries that have been proposed in the ordering phase.
\end{proof}

\begin{lemma} \label{lemma1}
No two correct members in an ordering booth agree on conflicting ordering IDs for the same data entry.
\end{lemma}

\begin{proof}
Under unchanged booths in the ordering phase, this lemma intuitively holds, as correct validators do not reply to a \texttt{Pre-Order} message that contains a previously assigned ordering ID (\circle<3> in \textcolor{black}{O}\ref{ordering:2}). Under dynamic booths, due to \algo's requirement on quorum constructions, each quorum must contain the proposer and pivot validator; if the proposer double assigns a data entry with two different ordering IDs and sends them to two ordering instances with different booths, one of the ordering instances will fail because the pivot validator must be included in both booths and does not collude with the proposer. When the pivot validator does not reply, the proposer cannot proceed with \textcolor{black}{O}\ref{ordering:3}, and this ordering instance will be aborted. Therefore, correct members do not agree on conflicting ordering IDs for the same data entry.
\end{proof}

\begin{theorem}[Safety] \label{safety}
All correct members agree on a total order for proposed data entries in the presence of less than $f$ failures in each booth.
\end{theorem}

\begin{proof}[Proof (sketch)]
We prove the safety theorem by contradiction. We claim that two correct members commit the two data entries in the same order. Say if two data entries are committed in consensus instances with the same order, then there must exist two quorums constructed in \textcolor{black}{C}\ref{c:c3} agreeing on the two entries, which have been appended to the total order log; otherwise, the proposer cannot receive sufficient votes in \textcolor{black}{C}\ref{c:c3}. In this case, the two entries must have been ordered with the same ordering ID in the ordering phase, which contradicts Lemma~\ref{lemma1}. Therefore, each entry is committed with a unique ordering ID.
\end{proof}

\begin{theorem}[Liveness]
A correct proposer eventually receives replies to proposed data entries in the consensus phase.
\end{theorem}

\begin{proof}
\algo assumes partial synchrony for liveness. After GST, the message delay and processing time are bounded. 
Since no two adjacent consensus instances leave uncommitted log entries in between, with Theorems~\ref{validity} and~\ref{safety}, a transaction that failed to be committed in a consensus instance will always be handled in a future consensus instance. Thus, during sufficiently long periods of synchrony, a correct proposer eventually receives replies from validators.
\end{proof}

\section{Gossiping}
\label{sec:gossiping}
The gossiping module strengthens system robustness by further disseminating committed transactions to the network. This module does not affect the correctness of consensus and can be applied based on preferences. When gossiping is enabled, after the consensus module commits a transaction, the proposer sends a \texttt{Gossip} message piggybacking the transaction to other connected vehicles, which will keep propagating this message to their connected vehicles. Therefore, the transmission forms a \emph{propagation tree} where the proposer is the root node named \emph{propagator} and other nodes are \emph{gossipers}.

Each \texttt{Gossip} message has a \emph{lifetime} (denoted by $\lambda$) that determines the number of propagators it traverses in a transmission link; i.e., the height of the propagation tree. For example, in Figure~\ref{fig:propagation-tree}, $V_1$ is the proposer and has direct connections with $V_2$, $V_3$, and $V_4$. The gossip message has a lifetime of $2$, and its propagation stops when a path has included $2$ propagators; e.g., path <$V_1$, $V_2$, $V_5$>. Specifically, we describe the gossiping workflow as follows.

\begin{SCfigure}[1][t]
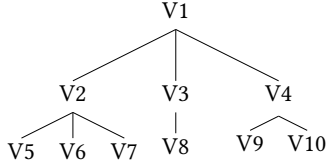

    \centering
    \begin{adjustbox}{width=0.5\linewidth}
    \Tree[.V1 [.V2 [.V5 ] [.V6 ] [.V7 ]]
          [.V3 [.V8 ]]
          [.V4 [.V9 ] [.V10 ]]
    ]
    \end{adjustbox}
    \caption{A propagation tree of a gossip message with $\lambda=2$. $V_1$ is the proposer and others are propagators.}
    \label{fig:propagation-tree}
\end{SCfigure}

\begin{description}
    \item[Proposer:] 
    The proposer ($V_p$) starts to disseminate committed transactions (the \texttt{Commit} message in \textcolor{black}{C}\ref{commitmsg}) by taking the following actions. 
    \begin{enumerate}[(a)]
            \item $V_p$ creates a lifetime (e.g., $\lambda = 3$) for a \texttt{Commit} message and a set $\mathcal{G}$ for recording the traverse information of the message's transmission path.
            
            \item $V_p$ hashes the \texttt{Commit} message and signs the combination of the hash result ($h_c$) and $\lambda$ to obtain a signature $\sigma_{V_p}$. Then, it adds a \emph{traverse entry} containing $\lambda$, $\sigma_{V_p}$, and $C_{V_p}$ to $\mathcal{G}$; i.e., $\mathcal{G}.add([\lambda, \sigma_{V_p}, C_{V_p}])$, where $C_{V_p}$ is $V_p$'s address and public key (\S~\ref{sec:boothcomposer}). 
            
            \item $V_p$ sends $\langle \texttt{Gossip}, \langle \texttt{Commit} \rangle, h_c, tx, \mathcal{G} \rangle$ to connected vehicles that are not included in the consensus process of the $\texttt{Commit}$ message and then creates a list for receiving \texttt{ack} messages from propagators.
    
        \end{enumerate}

    \item[Propagator:] A propagator ($V_i$) stores a valid \texttt{Gossip} message and may further disseminate the message if there is lifetime remaining.
    
    \begin{enumerate}[(a)]
        \item \label{gossip:propagator:a}
        $V_i$ verifies a received \texttt{Gossip} message by four criteria: \circle<1> it has not previously received this message; \circle<2> signatures in $\mathcal{G}$ are valid; \circle<3> $\lambda$s in $\mathcal{G}$ are strictly monotonically decreasing; and \circle<4> the message still has remaining lifetime; i.e., $\lambda_{min} = min \lbrace \mathcal{G.{\lambda}} \rbrace > 0$. If the verification succeeds, $V_i$ sends an \texttt{ack} message to this message's proposer and the pivot validator, registering itself on the propagator list.
        
        \item
        $V_i$ decrements this gossip message's lifetime; i.e., $\lambda_{new} {=} \max \lbrace \lambda_{min}-1, 0 \rbrace$. If $\lambda_{new} > 0$, then $V_i$ prepares to further disseminate the message. It signs the combination of $h_c$ and $\lambda_{new}$ to obtain a signature $\sigma_{V_i}$ and adds a new traverse entry into $\mathcal{G}$; i.e., $\mathcal{G}.add([\lambda_{new}, \sigma_{V_i}, V_i])$.
        
        \item $V_i$ sends $\langle \texttt{Gossip}, \langle \texttt{Commit} \rangle, h_c, tx, \mathcal{G} \rangle$ to other connected vehicles, excluding those it receives this \texttt{Gossip} message from (in 
        \textcolor{black}{(\ref{gossip:propagator:a})}).
    
    \end{enumerate}
    
\end{description}

The gossip module in \algo is an additional mechanism that provides an additional layer of redundancy. Since each propagator is registered to the proposer, the proposer can connect to propagators and read its own portion of data. This is especially important when consensus is conducted in small booth sizes, as the failure of vehicles can significantly affect the system's overall resilience.
Although gossiping does not guarantee the delivery of all transactions to all participants in the network, it can help to increase the probability of transaction delivery even in the face of network partitions or failures, strengthening the system's robustness without interfering with the core consensus process.

\section{The storage module}
\label{sec:storage}

\algo stores committed data batches in a \emph{storage master}. The storage master creates storage master instances (SMIs) for transactions from different vehicles, as a vehicle may operate in different roles in different booths (i.e., $\gamma > 1$). For example, in Figure~\ref{fig:vstorage}, the storage master has three instance zones containing proposer, validator, and gossiper SMIs (if the gossiping module is applied). A vehicle's storage master has at most one proposer SMI because it is the only proposer of its own \algo instance. When the vehicle joins other booths (operating as a validator), its storage master creates validator SMIs. The number of proposer and validator SMIs equals the vehicle's catering factor ($\gamma$). When gossiping is enabled, a vehicle may receive gossiping messages, and its storage master creates gossiper SMIs storing consensus results disseminated from other vehicles.

Each storage master instance has two layers: temporary and permanent, with the purpose of maximizing the usage of vehicles' limited storage space. Unlike the blockchain platforms working on servers (e.g., HyperLedger Fabric~\cite{androulaki2018hyperledger}, CCF~\cite{russinovich2019ccf}, and Diem~\cite{diem2020}), \algo operates among vehicles, which usually have only limited storage capability. 

The temporary storage temporarily stores transactions based on a predefined policy. \algo's implementation uses a timing policy that defines a time period (denoted by $\tau$) that the storage master stores a transaction (e.g., $\tau = 24$ hours). A transaction is registered in the temporary storage by default and is deleted after $\tau$ time if the user issues no further command. In addition, temporarily stored transactions can be moved to permanent storage and kept permanently per user request. This option hands over the control of storage to users. For example, when users suspect malfunctions of their vehicles, they may want to keep related transactions as evidence and move them to permanent storage. We introduce four SMI APIs for the layered design.

\begin{itemize}
\item \texttt{sm.RegisterToTemp(\&tx, time.Now())} is used to register a transaction to the temporary storage layer.

\item \texttt{sm.CleanUpTemp(time.Now())} is a daemon process that calculates the remaining time for transactions in the temporary storage and deletes expired ones.

\item \texttt{sm.MoveToPerm(\&tx)} is called by users, moving selected transactions from temporary to permanent storage.

\item \texttt{sm.DeletePerm(\&tx)} is called by users. A user may delete a permanently stored transaction after it has served the user's purpose.

\end{itemize}

The policy of temporary storage can be implemented differently. For example, a policy that defines a fixed size of storage space may apply. When the temporary storage exceeds the predefined size, it clears up old transactions in a fist-in-first-out (FIFO) manner. 

\begin{figure}[t]
    \centering
    \includegraphics[width=\linewidth]{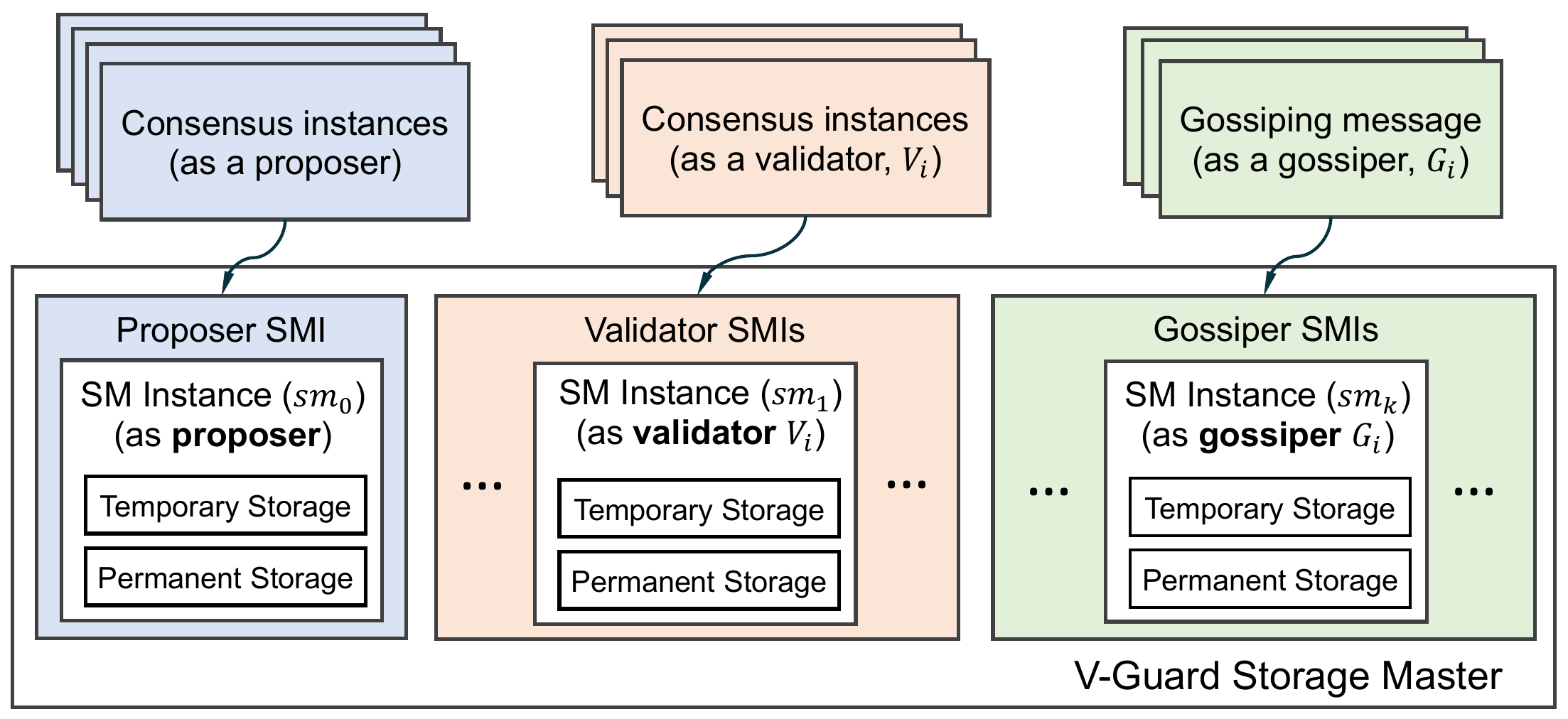}
    \caption{Storage master instances (SMIs). A vehicle may operate a proposer SMI when it conducts consensus for itself, validator SMIs when it participates in other vehicles' consensus, and gossiper SMIs when it enables the gossip module.}
    \label{fig:vstorage}
\end{figure}

The policy-based storage module can significantly reduce the use of storage space as vehicles often have limited storage capability. Note that the pivot validator can permanently store all data as automobile manufacturers often operate on their cloud platforms with scalable storage devices.

\begin{figure*}[h!]
\minipage{0.33\textwidth}
    \centering
    \includegraphics[width=0.99\textwidth]{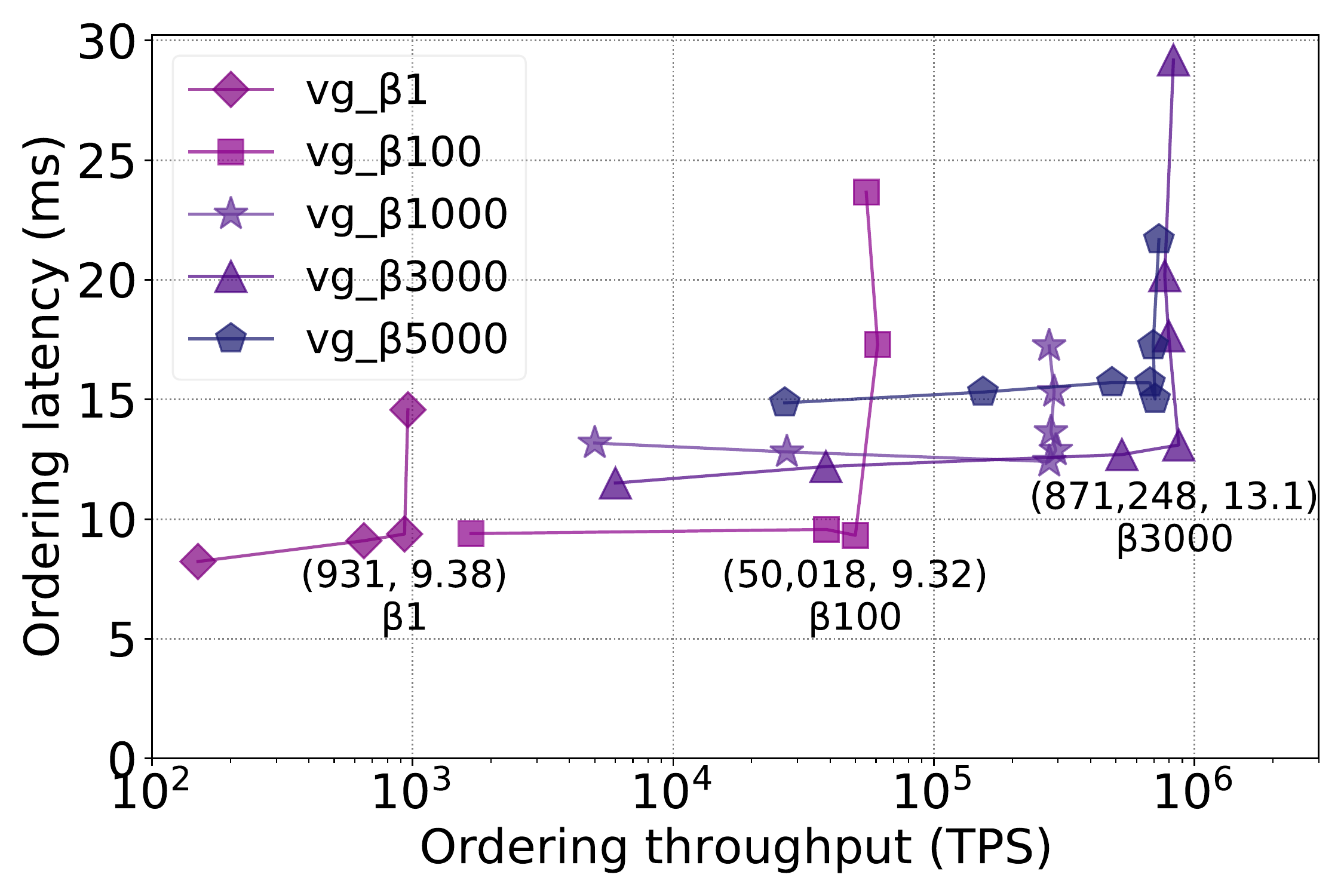}
    \caption{Ordering throughput vs. latency of \algo under varying batch sizes, where $\beta = 3000$ obtains the best performance under $n=4$, $\delta=0$, and $m=32$.}
    \label{fig:eval-vgordering-n4}
\endminipage \hfill
\minipage{0.33\textwidth}
    \centering
    \includegraphics[width=0.99\textwidth]{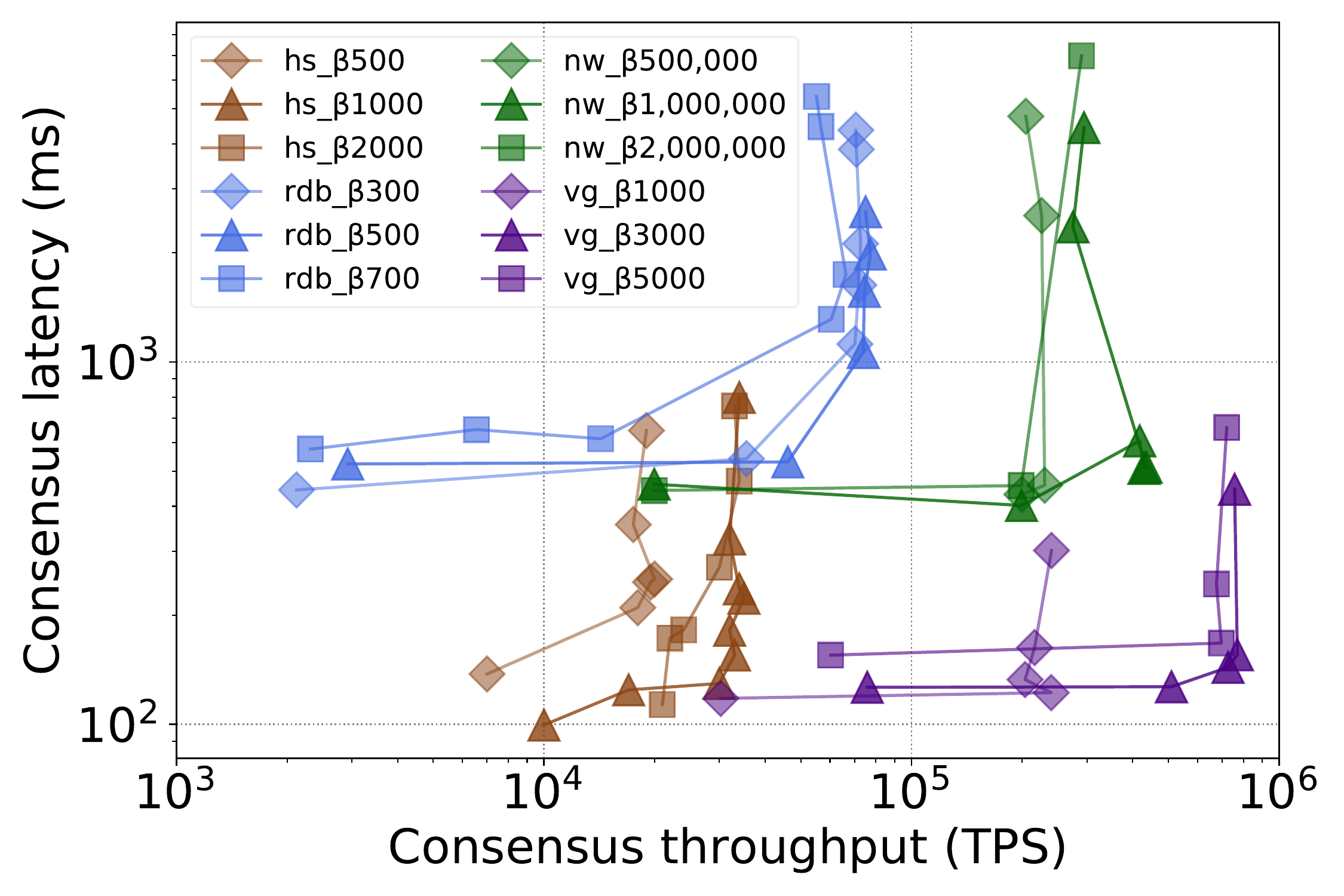}
    \caption{Consensus throughput vs. latency comparison of \algo and its baselines with varying batch sizes under $n=4$, $\delta=0$, and $m=32$.}
    \label{fig:eval-consensus-n4}
\endminipage  \hfill
\minipage{0.31\textwidth}
    \renewcommand{\arraystretch}{1.2}
    \begin{adjustbox}{width=\linewidth}
    \centering
        \begin{tabular}{rlrr}
        & Batch size  & { Throughput} & { Latency} \\
        &  & (TPS) & (ms)\\
        \hline
        HotStuff & ($\beta{=}1000$) & $34,015$ & 155\\
        ResilientDB & ($\beta{=}500$)  & $80,580$ & 4368\\
        Narwhal & ($\beta{=}10^{6*}$) & $423,058$ & $516$\\
        \textbf{V-Guard} (o) & ($\beta{=}3000$)  & $871,248$ & $13$ \\
         (c) &  & $765,930$ & $143$ \\
        \hline
        \end{tabular}
    \end{adjustbox}
    {
    \footnotesize
    * Narwhal uses fixed-bytes buffers as batches, so its buffer size $=m\times \beta_b$.
    }

    \caption{A summary of the peak performance of V-Guard (including ordering (o) and consensus (c)) and its baselines under their best batch sizes, where $n=4$, $\delta=0$, and $m=32$.}
    \label{fig:eval-table}
\endminipage
\end{figure*}

\section{Evaluation}
We compared the end-to-end performance of our \algo implementation (\texttt{vg}) against three state-of-the-art baseline approaches (using their open-source implementations): HotStuff~\cite{yin2019hotstuff, libhotstuff} (\texttt{hs}), a linear BFT protocol, whose variant is used in the Facebook Diem blockchain~\cite{diem2020}; ResilientDB~\cite{gupta13resilientdb, rdbgithub} (\texttt{rdb}), a BFT key-value store using PBFT~\cite{castro1999practical} as its consensus protocol; and Narwhal~\cite{danezis2022narwhal, nwgithub} (\texttt{nw}), a DAG-based mempool protocol that distributes transactions before consensus. We deployed them on $4$, $16$, $31$, $61$, and $100$ VM instances on a popular cloud platform~\cite{canadacloud}. Each instance includes a machine with $2$ vCPUs supported by $2.40$ GHz Intel Xeon processors (Skylake) with a cache size of $16$ MB, $7.5$ GB of RAM, and $90$ GB of disk space running on Ubuntu $18.04.1$ LTS. The TCP/IP bandwidth measured by \texttt{iperf} and the raw network latency between two instances are around $400$ Megabytes/s and $2$~ms, respectively. 

We introduce the following notation for reporting the results:

\begin{tabular}{@{}rp{7.2cm}}
    $m$ & The message size (bytes); i.e., the size of a data entry sent to achieve consensus, excluding any message header (other parameters in the message).\\
    $\beta$ & The batch size; i.e., the number of data entries (transactions) in a data batch sent to achieve consensus.\\
    $n$ & The number of nodes (the membership/cluster size). \\
    $\delta$ & Emulated network delay implemented by \texttt{netem}.\\
\end{tabular}

\textbf{Network delays.} We used \texttt{netem} to implement additional network delays of $\delta=0, 10\pm5, 50\pm10, 100\pm20$~ms in normal distribution at all scales. The performance change under the emulated network delays implies the performance in more complicated networks (e.g., wireless networks).

\textbf{Workloads.} \algo is a versatile blockchain that enables users to define their own message types. The size of messages is a critical factor in replication as all messages need to be serialized for transmission over the network, regardless of their type. To ensure a fair comparison with other baselines, the evaluation presented in the paper focuses on the replication performance of plain messages with different message sizes of $m=32$, $64$, and $128$ bytes. However, it is worth noting that \algo's messaging service is designed to be easily adaptable to support different types of messages. Some projects have already leveraged this feature and developed their own message types, which can be found on the \algo GitHub page~\cite{vguardsource}.

The rest of this section is organized as follows. \S~\ref{sec:eval:stationary} shows the evaluation result of \algo comparing against its baselines in a static membership; \S~\ref{sec:eval:dynamic} reports the performance of \algo's unique feature operating under dynamic memberships; and \S~\ref{sec:eval:summary} summarizes the evaluation results.

\subsection{Performance in static memberships}
\label{sec:eval:stationary}

We set \algo's consensus interval time to $\Delta=100$~ms; i.e., a consensus instance is scheduled every $100$~ms. Performance was reported in throughput and latency, where throughput is calculated as transactions per second (TPS) (i.e., the number of entries each algorithm commits per second), and latency was measured as the time elapsed between sending a request and receiving confirmation.

Batching is applied to measure performance. Under increasing batch sizes, throughput gains a diminishing marginal increase and peaks at a particular (the best) batch size; after this, throughput decreases as batching becomes more costly than other factors in the consensus process. 

To make fair comparisons, \algo operates only one instance (i.e., $\gamma =1$) when comparing against its baselines in \S\ref{sec:eval-sta-peak} and~\ref{sec:eval-sta-scalability}; the performance of multiple simultaneously running \algo instances is shown \S\ref{sec:eval:sta-catering}.

\begin{figure*}[t]
\minipage{\textwidth}
\centering
\includegraphics[width=0.5\textwidth]{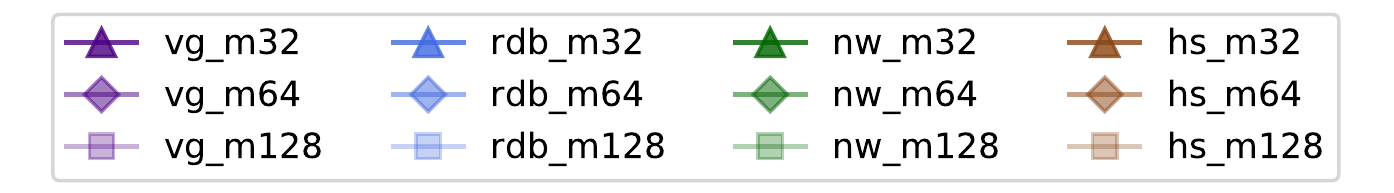}\endminipage\\
\minipage{0.49\textwidth}
    \includegraphics[width=\textwidth]{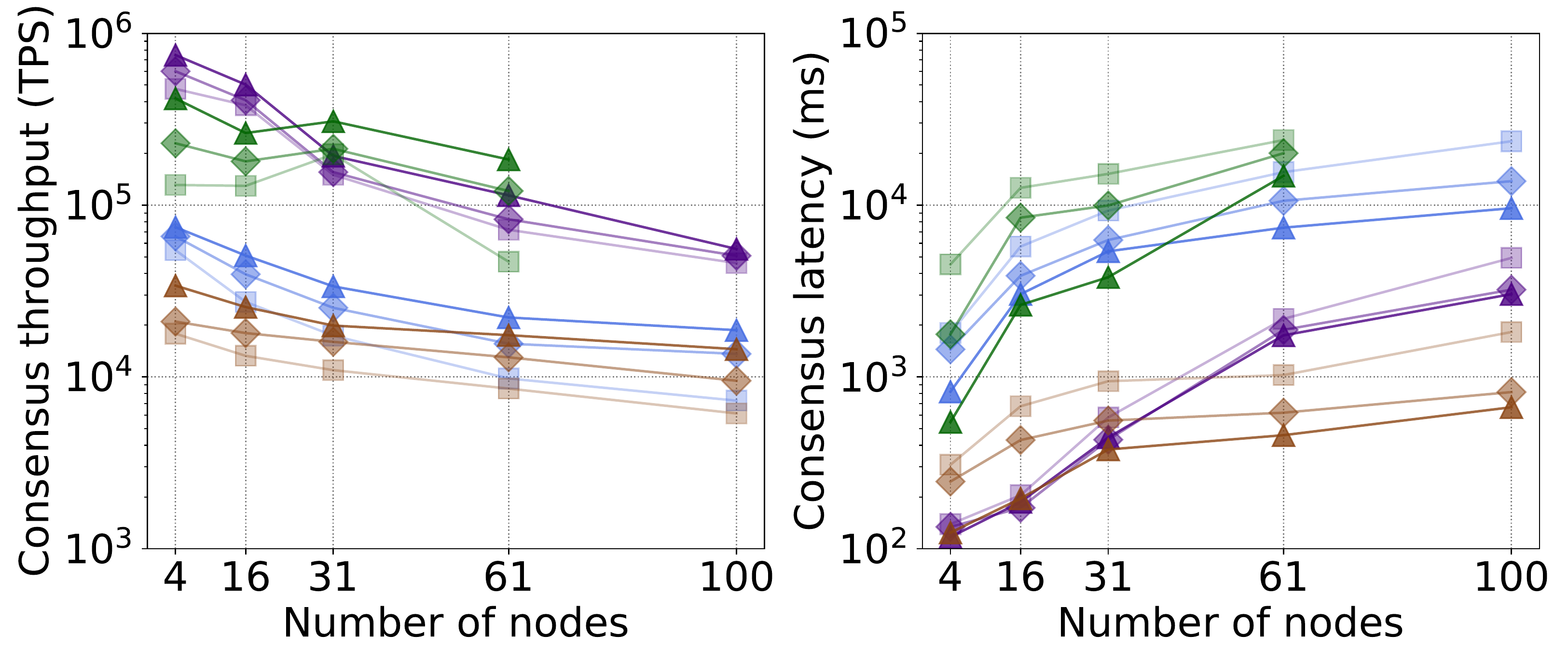}
    \subcaption{$\delta = 0$ ms.}
    \label{fig:eval-scala-l0}
\endminipage \hfill
\minipage{0.49\textwidth}
    \includegraphics[width=\textwidth]{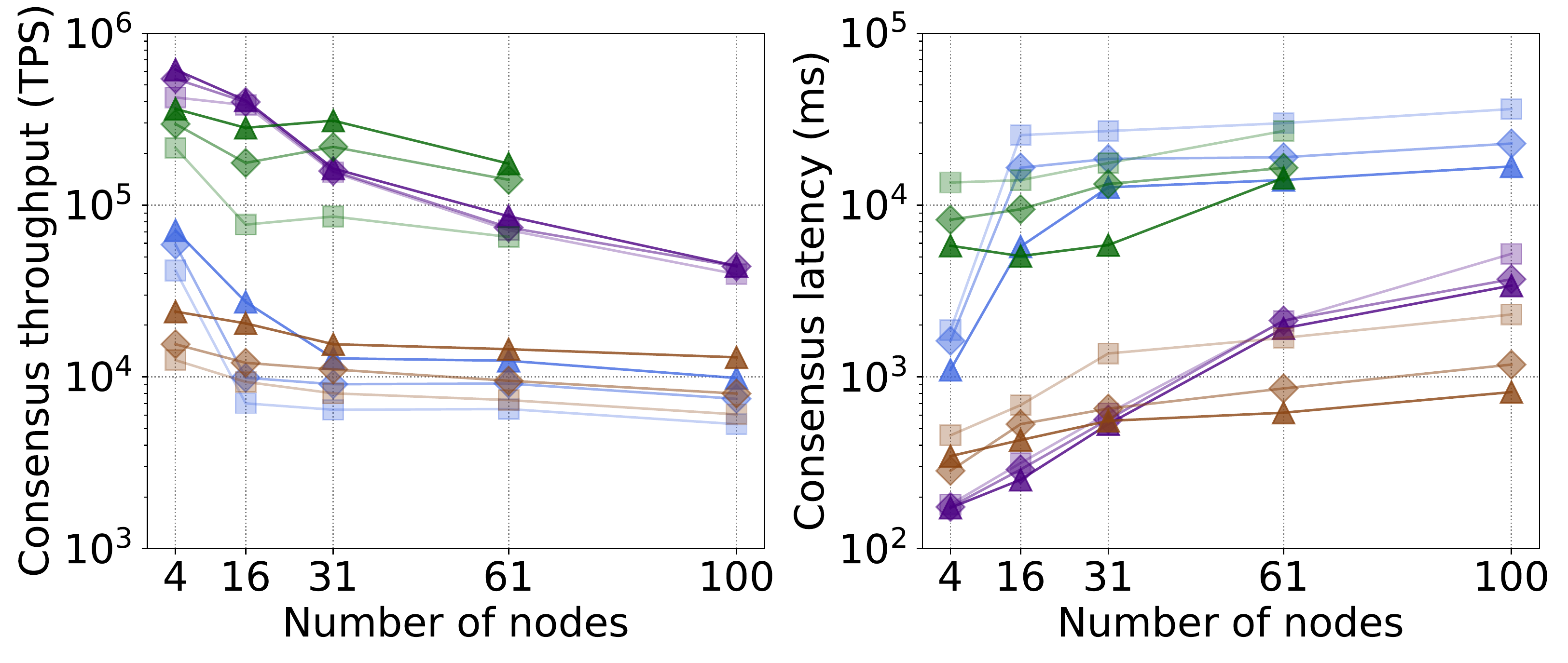}
    \subcaption{$\delta = 10\pm5$ ms.}
    \label{fig:eval-scala-l10}
\endminipage\\
\minipage{0.49\textwidth}
    \includegraphics[width=\textwidth]{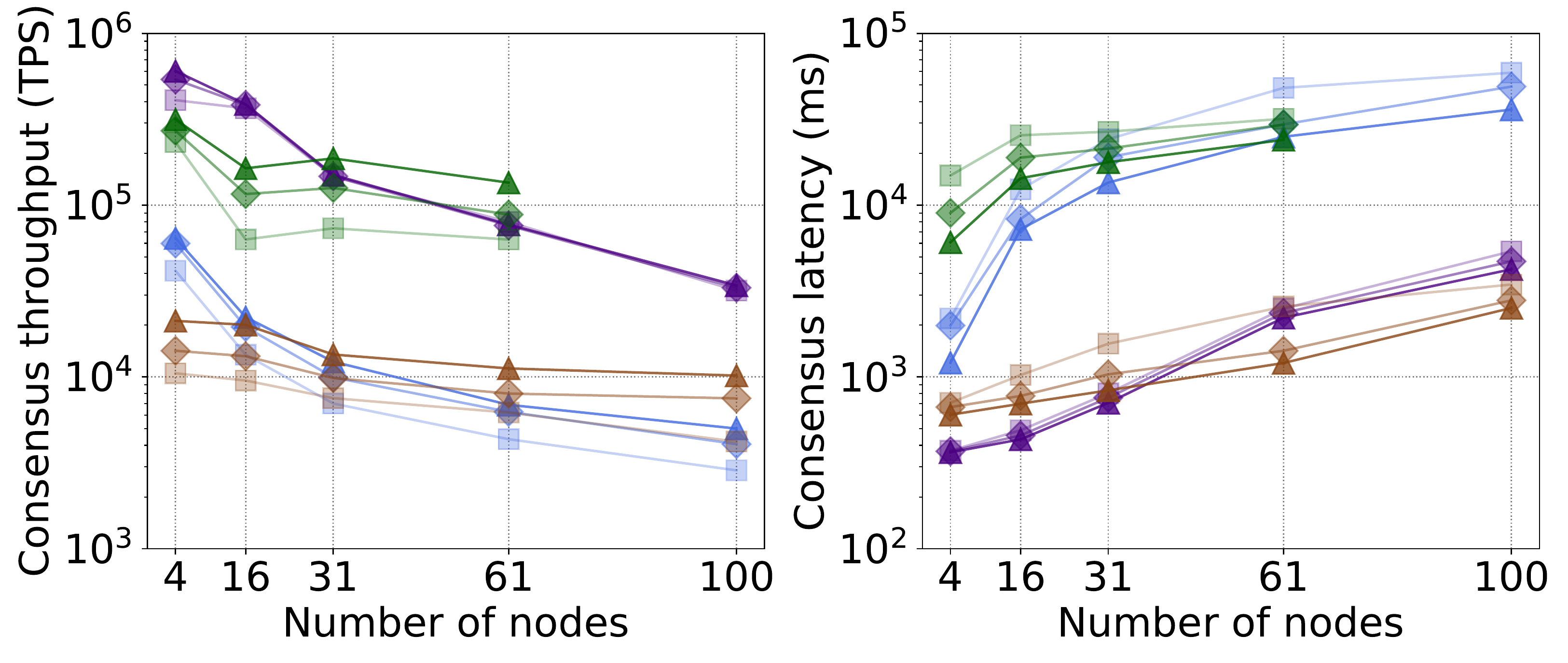}
    \subcaption{$\delta = 50\pm10$ ms.}
    \label{fig:eval-scala-l50}
\endminipage \hfill
\minipage{0.49\textwidth}
    \includegraphics[width=\textwidth]{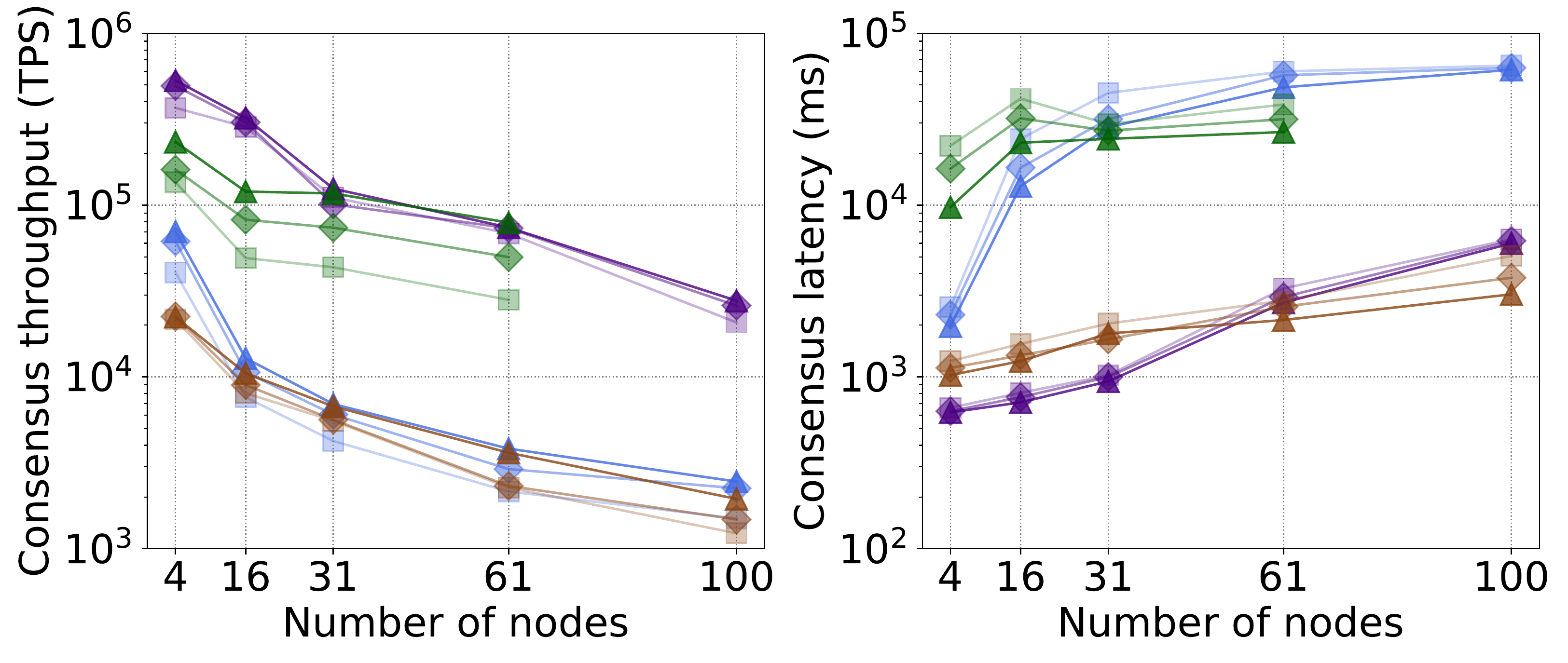}
    \subcaption{$\delta = 100\pm20$ ms.}
    \label{fig:eval-scala-l100}
\endminipage\\
\caption[Caption for LOF]{Performance comparisons of \algo and its baselines using their best batch sizes with varying message sizes of $m=32$, $64$, and $128$ bytes under implemented network delays of $\delta=0$, $10\pm5$, $50\pm10$, and $100\pm20$~ms (normal distribution) at increasing system scales of $n=4$, $16$, $31$, $61$, and $100$ nodes. (Narwhal's mempool protocol was unable to initialize because its required memory space exceeded our system resources under $n=100$.\protect\footnotemark
)}
\label{fig:eval:scalability}
\end{figure*}

\subsubsection{Peak performance}
\label{sec:eval-sta-peak}
The peak performance of all approaches occurred at $n{=}4$ nodes under a message size of $m{=}32$ bytes without implemented network delay (i.e., $\delta=0$~ms). To obtain \algo's peak performance, we kept increasing its batch size and monitoring the result (shown in Figure~\ref{fig:eval-vgordering-n4}): the ordering throughput first increases with diminishing marginal returns (e.g., $931>\frac{50,018}{100}=500.18 > \frac{871,248}{3000}=290$) and peaks at a throughput of $871,248$~TPS with a latency of $13$~ms under $\beta=3000$. After that, under a higher $\beta$, the ordering throughput decreases while the latency increases. We used this method to measure the peak performance of consensus throughput and latency for all approaches (shown in Figure~\ref{fig:eval-consensus-n4}). Under each batch size, we kept increasing clients to saturate the network (until an elbow of a curve occurs). We show the result of \algo and its baselines under three batch sizes that result in near-optimal (e.g., \texttt{vg\_1000}), optimal (e.g., \texttt{vg\_3000}), and over-optimal performance (e.g., \texttt{vg\_5000}). 

Figure~\ref{fig:eval-table} summarises the peak performance of all approaches, where \algo outperforms its baselines both in throughput and latency. Since consensus instances periodically commit ordered data entries, ordering and consensus have similar throughput.
\algo's high performance benefits from its implementation and design. Compared with the implementations of HotStuff~\cite{libhotstuff} and ResilientDB~\cite{rdbgithub}, \algo makes use of threshold signatures and obtains a linear message complexity. In addition, the separated ordering reduces latency by issuing ordering instances concurrently, and shuttled consensus significantly reduces message passing, which saves network bandwidth and increases throughput.

\subsubsection{Scalability}
\label{sec:eval-sta-scalability}

We then measured the consensus performance of \algo and its baselines at increasing message sizes and system scales under varying network delays (shown in Figure~\ref{fig:eval:scalability}). We set each approach's batch size to their best (optimal) batch sizes under $n=4$. The results show that when the message size and network delay increase, all approaches experience a decline in throughput and an increase in latency; in particular, message sizes have a more pronounced effect on throughput while network delays have a more direct effect on consensus latency.

\begin{figure*}[t]
\minipage{0.49\textwidth}
    \includegraphics[width=\textwidth]{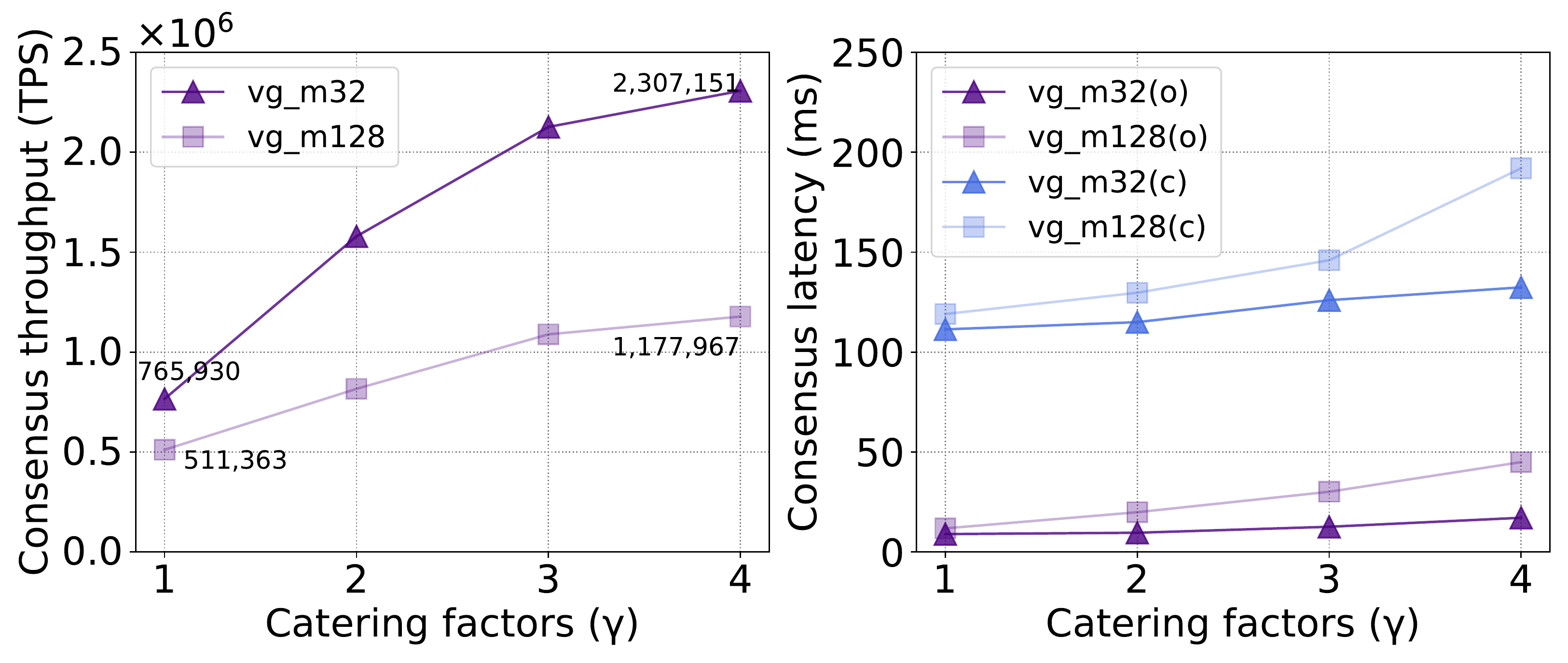}
    \subcaption{ $n=4$.}
    \label{fig:eval-catering-n4}
\endminipage \hfill
\minipage{0.495\textwidth}
    \includegraphics[width=\textwidth]{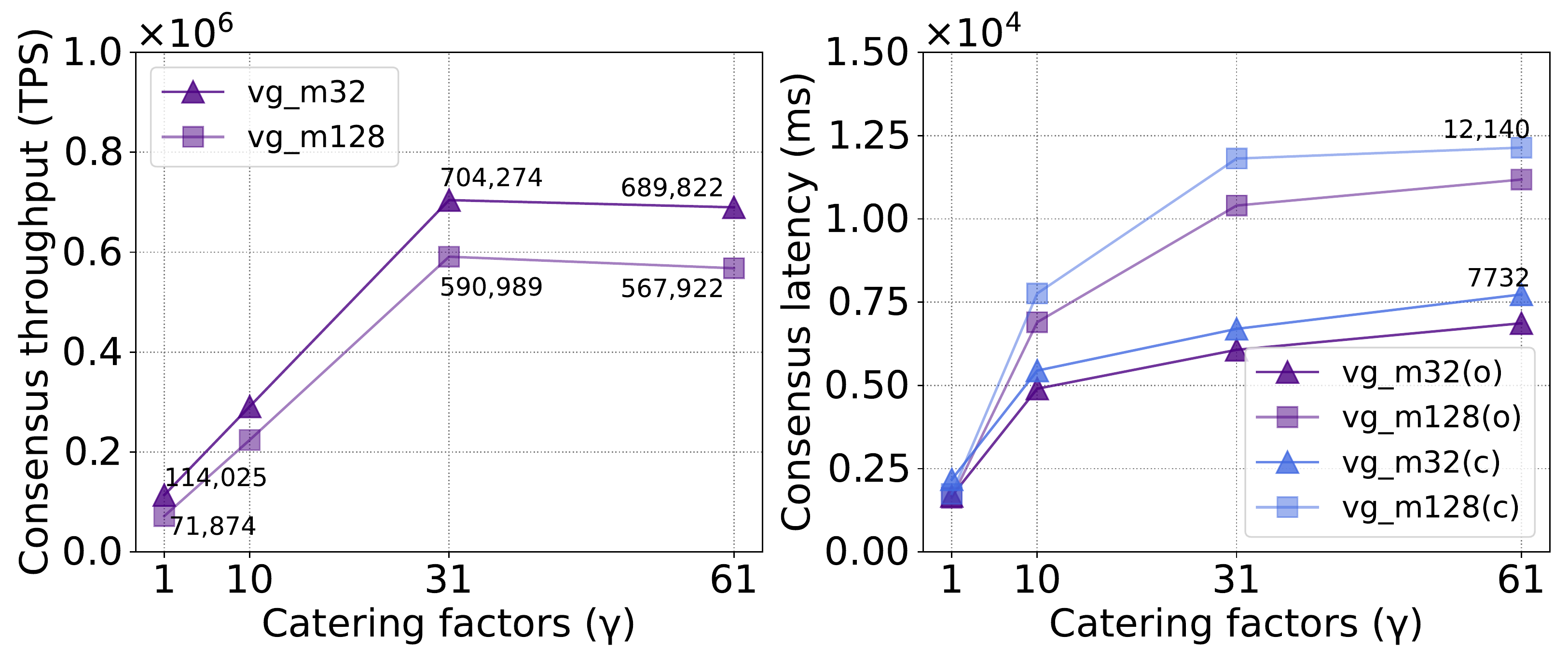}
    \subcaption{$n=61$.}
    \label{fig:eval-catering-n61}
\endminipage
\caption{\algo's performance under increasing catering factors, where $\delta=0$, $m=32$, and $\beta = 3000$.}
\label{fig:eval:catering}
\end{figure*}

\begin{figure*}[t]
\minipage{0.49\textwidth}
    \includegraphics[width=\textwidth]{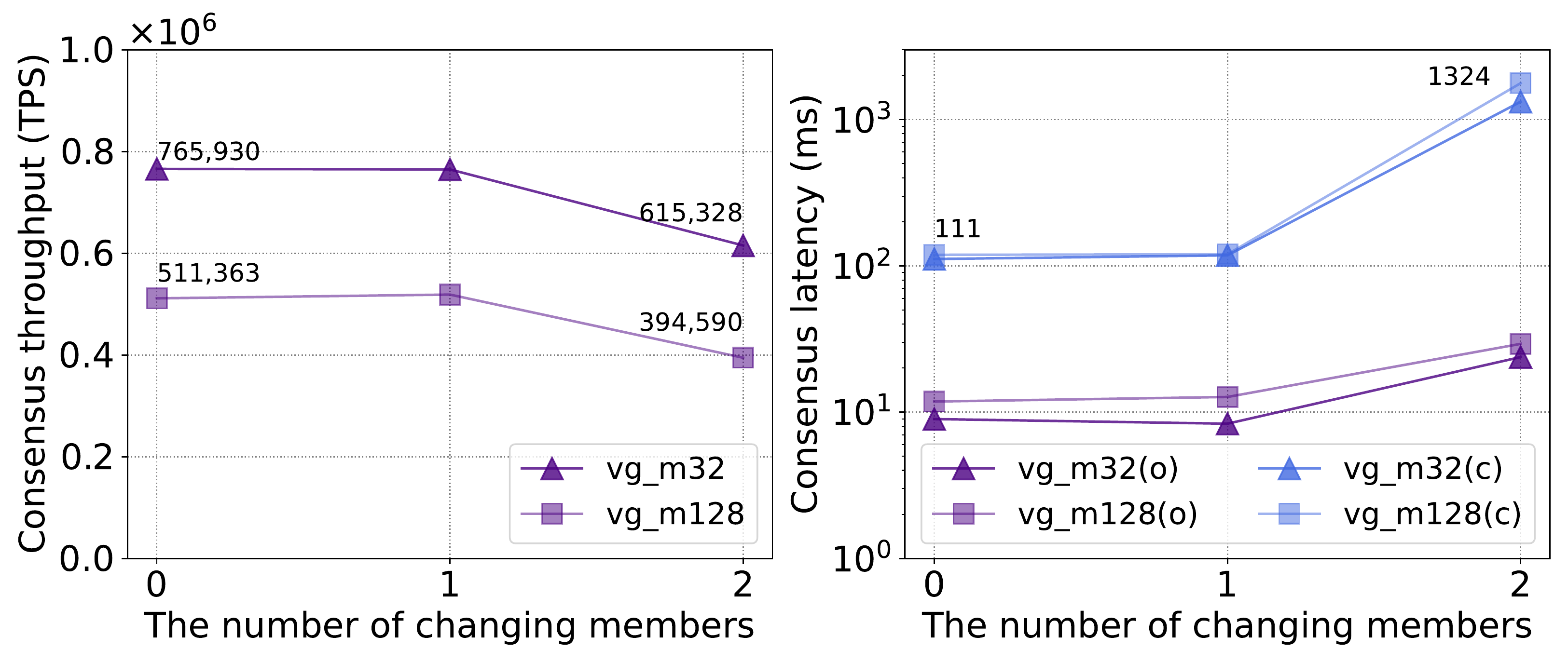}
    \subcaption{ $n=4$.}
    \label{fig:eval-dynamic-n4}
\endminipage \hfill
\minipage{0.49\textwidth}
    \includegraphics[width=\textwidth]{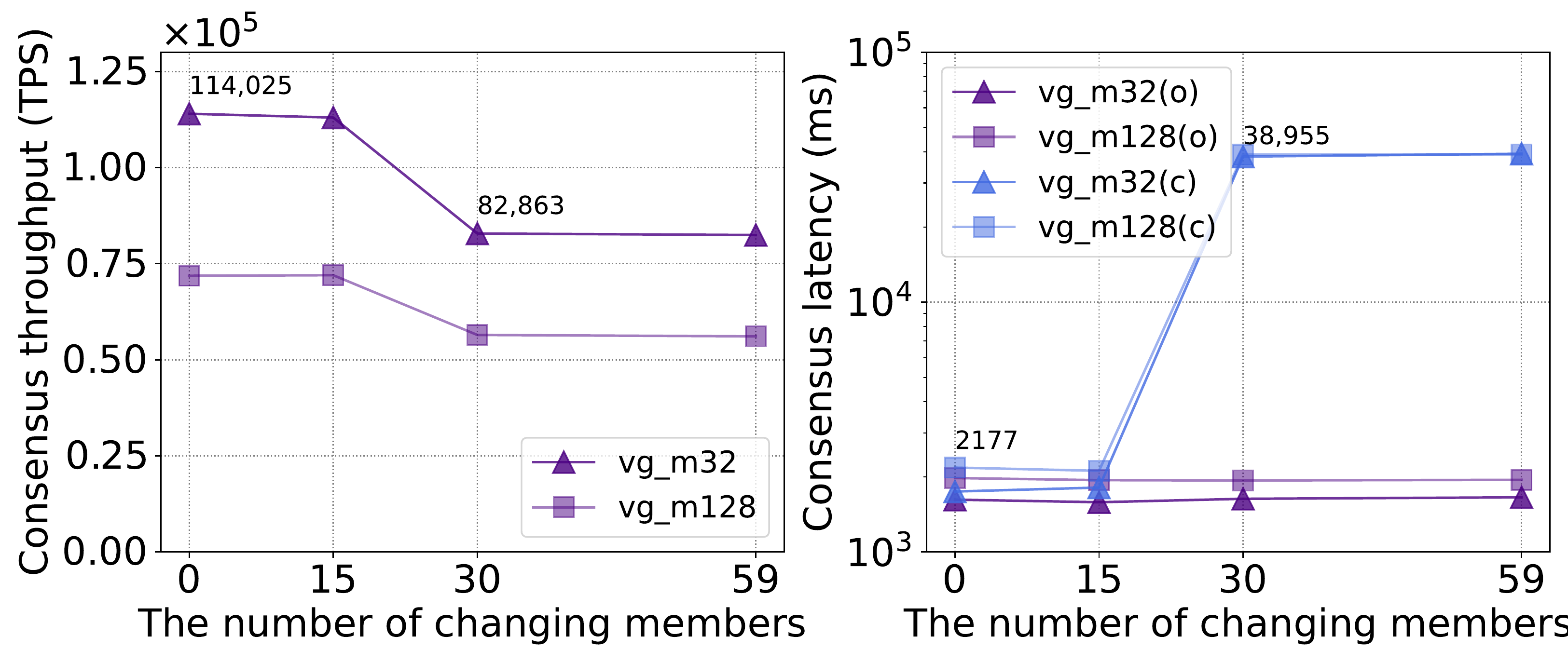}
    \subcaption{$n=61$.}
    \label{fig:eval-dynamic-n61}
\endminipage
\caption{\algo's performance under dynamic memberships, where $\delta=0$ and $\beta = 3000$. Since a \algo membership always includes the proposer and pivot validator, the maximum changing members allowed under $n=4$ and $61$ are $2$ and $59$ nodes, respectively.}
\label{fig:eval-dynamic-mem}
\end{figure*}

Under increasing system scales, \algo, HotStuff, and ResilientDB, which use a single leader to manage the consensus and transaction distribution, witness a decrease in throughput. In contrast, Narwhal observes an increase in throughput when the system scales up from $16$ to $31$ nodes. Narwhal develops a mempool protocol that distills the distribution of transactions from consensus. Each node in Narwhal runs as a worker and a validator, where the worker disseminates transactions by the DAG-based mempool protocol, while the validator concentrates on achieving consensus on the hash of transactions. Consequently, Narwhal's throughput increases (from $16$ to $31$ nodes) when more transactions are fed into the network but then starts to decline when the transaction distribution becomes a heavier burden than the consensus process. While Narwhal exhibits impressive throughput, its DAG structure results in high latency, making Narwhal one of the highest latency algorithms. 

\algo manages to achieve the highest throughput at small scales. This is due to the separation of the ordering and consensus processes, which significantly reduces message passing overhead. As the system scales up, \algo, despite only having one node (the proposer) to disseminate data entries, still maintains a comparable throughput to Narwhal, even though Narwhal uses a DAG structure and has multiple nodes disseminating data entries. Overall, \algo's separation of ordering and consensus processes allows it to achieve high throughput with low latency, even at large scales.

In addition, \algo shows the mildest performance drop under increasing message workloads, system scales, and network delays among all approaches. 
In particular, compared to its baselines, \algo's consensus throughput and latency are the least affected by higher network delays.
For example, when $n=4$ and $m=32$, its throughput drops from $765,930$~TPS under $\delta=0$~ms to $531,594$~TPS under $\delta=100\pm20$~ms while latency rises from $143$ to $622$~ms. 
 \algo's ordering instances take place in a non-blocking manner, without waiting for the completion of their prior instances, and shuttled consensus instances amortize the network delay among data entries included in $\Delta$. 
This feature makes \algo a strong fit for V2X applications, where higher network latency is the norm.

\footnotetext{{Narwhal's mempool protocol requires static memory allocations based on the system and batch sizes; when $n{=}100$, it requires each worker to allocate $5.7$ Gigabytes of static memory for the protocol alone before the program runs, which cannot be accommodated by our system resources.}}

\subsubsection{Performance under increasing catering factors}
\label{sec:eval:sta-catering}

We also evaluated \algo's performance when each node operates multiple \algo instances (i.e., $\gamma > 1$). For example, in a $4$-node system ($n_1$ to $n_4$) with $\gamma = 2$, $n_1$ and $n_2$ start their individual \algo instances and operate as the proposer; $n_1$ also operates as a validator in $n_2$'s instance and vice versa. The remaining two nodes, $n_3$ and $n_4$, operate as a validator in both $n_1$'s and $n_2$'s instances.

We measured \algo's throughput and latency at two system scales under increasing catering factors from $\gamma=1$ to $\gamma=n$ (shown in Figure~\ref{fig:eval:catering}). When $n{=}4$, the throughput increases with diminishing gains and peaks at a throughput of 2,307,151 TPS, while both the ordering and consensus latencies experience a slight increase. When $n{=}61$ and $m=32$, throughput peaks at $\gamma=31$  with $704,274$ TPS and decreases at $\gamma=61$ with $689,822$ TPS. With a higher catering factor, the performance of each instance drops as system resources are amortized among all instances.

\subsection{Performance in dynamic memberships}
\label{sec:eval:dynamic}

To demonstrate \algo's capability to achieve consensus under frequently changing membership, we assessed its performance under the worst-case scenario, wherein each data entry is processed and committed in separate booths. This means that membership alterations occur at the same rate as transaction consensus.

We measured the performance operating under dynamic memberships at two scales: $n{=}4$ and $n{=}61$. Since \algo always includes the proposer and pivot validator in membership, it allows a maximum of $2$ and $59$ members (nodes) to change under $n{=}4$ and $n{=}61$, respectively. In this case, each ordering and consensus instance gets a new booth from the MMU. For example, when $n{=}4$ with $2$ changing member, the ordering and consensus booths have one different member:

\vspace{0.5em}
\begin{tabular}{@{}rl}
    \vspace{0.1em}
    Ordering booth: & [$n_1$ (as $V_p$), $n_2$ (as $V_{\pi}$), $n_3$, and $n_4$]\\
    Consensus booth: & [$n_1$ (as $V_p$), $n_2$ (as $V_{\pi}$), \textcolor{blue}{$n_5$}, and \textcolor{blue}{$n_6$}]
    \vspace{0.1em}
\end{tabular}
\vspace{0.5em}

When $n{=}61$, we evaluated the performance of throughput and latency under $15$, $30$, and $59$ changing members under two messaging workloads ($m{=}32$ and $128$ bytes).

The results show that dynamic memberships do not affect performance when the number of changing members is less than $f$. Since each ordering and consensus instance takes place in different booths, when a consensus instance includes new members that have not previously seen the ordered data entries, new members must undergo a time-consuming task to validate all signatures of data entries included in a transaction. 
Thus, old members reply to the proposer faster than the new members; when the remaining old members can still form a quorum ($2f+1$), the system performance is not affected by changed members; e.g., $1$ and $15$ changing members in Figure~\ref{fig:eval-dynamic-n4} and Figure~\ref{fig:eval-dynamic-n61}, respectively.

In contrast, if the number of remaining old members in a consensus quorum drops below $2f+1$, the quorum will have to wait for new members to join. This can significantly impact performance, as new members become the bottleneck. The consensus latency increases as new members must complete the validation task and respond to the proposer, leading to a drop in overall performance. When $m=32$, throughput drops by $19\%$ and $27\%$ under $n=4$ and $n=61$, respectively. With $30$ new members in $n=61$, the consensus latency increases from $2177$ to $38,955$~ms.

\subsection{Summary of results}
\label{sec:eval:summary}
Under static membership, \algo outperforms its baselines in terms of throughput and latency (\S~\ref{sec:eval-sta-peak}). When the system scales up, \algo shows a similar performance decrease as its baselines under varying workloads (\S~\ref{sec:eval-sta-scalability}). Under increasing network delays, \algo exhibits a lower decrease in throughput than its baselines at all scales (\S~\ref{sec:eval-sta-scalability}). 
With increasing \algo instances ($\gamma >1$), throughput first increases to a peak and then starts to decrease (\S~\ref{sec:eval:sta-catering}). Under dynamic membership, \algo's performance remains unchanged when new members are fewer than $f$; with more than $f$ new members, latency surges drastically as the newly joined members must undergo the verification process (\S~\ref{sec:eval:dynamic}).

\section{Related work}
\label{sec:relatedwork}
\textbf{Blockchains in V2X networks.} 
V2X blockchains should be able to achieve consensus in real-time for a large volume of data, such as vehicle status data generated by Telsa Autopilot~\cite{autopilot} and Cadillac SuperCruise~\cite{cadillacsupercurise}, and cope with unstable connections among vehicles on the roads~\cite{ibmblockchain, limechain}.
Depending on whether membership management is required, blockchains can be categorized as permissionless and permissioned blockchains. Permissionless blockchains, applying the Proof-of-X protocol family (e.g., Proof-of-Work~\cite{nakamoto2019bitcoin} and Proof-of-Stake~\cite{gilad2017algorand}) to reach probabilistic agreement, allow participants to join anonymously. However, due to their limited performance, permissionless blockchains are often not deployed in latency-critical V2X blockchain applications coping with a large volume of data in real-time~\cite{hassija2020dagiov, elagin2020technological, yuan2016towards, jiang2018blockchain}, despite some privacy-critical V2X use cases because they require anonymous participation~\cite{peng2021privacy, feng2019survey}. 

In contrast, permissioned blockchains, registering participants with identities, achieve high throughput and low latency where participants can be verified by their signatures. They often apply BFT consensus algorithms (e.g., PBFT~\cite{castro1999practical}, SBFT~\cite{gueta2019sbft} and HotStuff~\cite{yin2019hotstuff}) to achieve deterministic agreement among all participants and have been favored by pioneering blockchain testbeds initiated by major automobile manufacturers~\cite{bmwbc, mercedesbc, vwbc1, toyotabc}. However, they are not adaptive to dynamic memberships; when membership changes, consensus has to temporarily stop and facilitates for updating configuration profiles~\cite{rodrigues2010automatic}.

\textbf{BFT consensus algorithms.}
Consensus algorithms are at the core of providing safety and liveness for state machine replication (SMR)~\cite{schneider1990state}. BFT SMR provides an abstraction of consensus services that agree on a total order of requests in the presence of Byzantine failures~\cite{schneider1984byzantine, lamport1982byzantine, lamport2019byzantine, attiya2004distributed}. Byzantine failures are becoming more common in blockchains as participating users may intentionally break the protocol (launching Byzantine attacks) to gain more profit~\cite{zhang2020byzantine, daian2020flash, lewis2014flash}.

Leader-based BFT algorithms have been widely used by permissioned blockchains, such as HyperLedger Fabric~\cite{androulaki2018hyperledger} and Diem~\cite{diem2020}. After PBFT~\cite{castro1999practical} pioneered a practical BFT solution, numerous approaches have been proposed for optimizations from various aspects. They use speculative decisions and reduce workloads for the single leader~\cite{kotla2007zyzzyva, duan2014hbft, gunn2019making}, 
develop high performance implementations~\cite{el2019blockchaindb, bessani2014state, guerraoui2010next, buchman2016tendermint, sousa2018byzantine, qi2021bidl}, 
reduce messaging costs~\cite{song2008bosco, yang2021dispersedledger, distler2011increasing, martin2006fast, liu2016xft, neiheiser2021kauri}, obtains linear message complexity by using threshold signatures~\cite{gueta2019sbft, yin2019hotstuff, zhang2021prosecutor}, 
distill transaction distribution from consensus~\cite{danezis2022narwhal}, 
limit faulty behavior using trusted hardware~\cite{behl2017hybrids, chun2007attested, kapitza2012cheapbft, levin2009trinc, decouchant2022damysus}, 
improve fault tolerance thresholds~\cite{hou2022using, xiang2021strengthened}, offer confidentiality protection dealing with secret sharing~\cite{vassantlal2022cobra},
and apply accountability for individual participants~\cite{civit2021polygraph, shamis2022ia, neu2021ebb}. 

In addition to leader-based algorithms, leaderless BFT algorithms avoid single points of failure and single leader bottlenecks~\cite{lamport2011brief, miller2016honey, crain2018dbft, duan2018beat, suri2021basil}. Without a leader, leaderless BFT algorithms often utilize binary Byzantine agreement~\cite{mostefaoui2014signature} to jointly form quorums~\cite{ben1994asynchronous} but suffer from high message and time costs for conflict resolutions.

\textbf{Membership reconfiguration.} Membership changes involve reconfiguration of consensus systems~\cite{lamport2010reconfiguring}, where the agreement protocol needs to switch to the updated replica-group configuration~\cite{distler2021byzantine}. 
Under loosely-consistent models, some approaches implement peer-to-peer lookups that determine a system membership as the neighborhood of a certain identifier (e.g, \cite{stoica2001chord, castro2002secure}). Under strong-consistent models, a sequence of consistent views of the system memberships is required to ensure the safety of consensus processes~\cite{johansen2006fireflies, cowling2009census, guerraoui2001generic}. For example, MC~\cite{rodrigues2010automatic} proposes a membership reconfiguration approach for PBFT~\cite{castro1999practical}, keeping track of system membership and periodically notifying other system nodes of membership changes. However, transaction consensus must wait for each membership reconfiguration to be completed, which makes it inefficient during frequent membership changes.

\section{Conclusions}
\algo is a permissioned blockchain that efficiently achieves consensus when system memberships constantly change (dynamic environment). It binds each data entry with a system membership profile and obtains chained consensus results of both proposed data and membership profiles. 
In addition, it separates ordering from consensus and allows data entries to be ordered and committed among different members. This separation brings high performance both in static and dynamic environments. Under static memberships, \algo's peak throughput is $1.8\times$ higher than its best-performing baseline with the lowest consensus latency. In a dynamic environment, \algo exhibits flexibility to operate under changing memberships achieving $73\%$ of the throughput of its performance in static environments.

\bibliographystyle{ACM-Reference-Format}
\bibliography{ref}

\end{document}